\documentclass{article}
\usepackage{graphicx}
\usepackage[left=2cm]{geometry}
\usepackage[page]{appendix}
\usepackage{amsmath}
\usepackage{amssymb}
\usepackage{amsthm}
\usepackage{slashbox}
\usepackage{stmaryrd}
\usepackage{txfonts}
\usepackage{pstricks,pst-node}
\usepackage{caption}
\usepackage{subcaption}
\usepackage{url}
\usepackage{multirow}
\newtheorem{lemma}{Lemma}[section]

\newtheorem{proposition}[lemma]{Proposition}

\newtheorem*{theorem*}{Theorem}

\DeclareMathOperator{\Var}{Var}

\begin{document}

\title{Prior Distributions for Ranking Problems} 
\author{Toby~Kenney \and Hao~He \and Hong Gu}

\maketitle

\begin{abstract}
The ranking problem is to order a collection of units by some
unobserved parameter, based on observations from the associated
distribution. This problem arises naturally in a number of contexts,
such as business, where we may want to rank potential projects by
profitability; or science, where we may want to rank variables
potentially associated with some trait by the strength of the
association. Most approaches to this problem are empirical Bayesian,
where we use the data to estimate the hyperparameters of the prior
distribution, then use that distribution to estimate the unobserved
parameter values. There are a number of different approaches to this
problem, based on different loss functions for mis-ranking
units. However, little has been done on the choice of prior
distribution. Typical approaches involve choosing a conjugate prior
for convenience, and estimating the hyperparameters by MLE from the
whole dataset. In this paper, we look in more detail at the effect of
choice of prior distribution on Bayesian ranking. We focus on the use
of posterior mean for ranking, but many of our conclusions should
apply to other ranking criteria, and it is not too difficult to adapt
our methods to other choices of prior distributions.
\end{abstract}

\subsection*{keywords}
Empirical Bayes; posterior mean ranking; choice of prior

\section{Introduction}

Suppose we have a collection of units we want to rank by a certain
feature of each unit: for example, we may wish to rank genes by the
risk they cause of a particular condition; we may wish to rank
sportsmen by their success-rate at particular standardised trials; we
may wish to rank business opportunities by the profit they will
generate. This is a very common inference problem first studied as a
formal statistical problem by Bechhofer (1954) and by Gupta
(1956). Typically, for each unit we wish to rank, we will have some
data on the associated feature, but will not know the true value of that
feature. Based on our data, we will have a point estimate for the
feature, and an associated error distribution. The amount of data we
might have for different units can vary wildly, meaning that the
associated error distributions can be very different for different
units. This means that when we select the top units using only our
point estimates, the units for which we have largest errors have a
higher chance of appearing among the top units, because a large error
increases the chance of the point estimate being large. We are
therefore likely to select a large number of false positives if we
select based solely on the point estimates.

We can illustrate this with a simple example. Suppose we
have 300 coins, we toss 100 of them six times each, 100 of them eight
times each, and the remaining 100 of them ten times each, and rank the
coins by the proportion of heads observed. If the coins are all fair,
then among the 100 that we toss six times each, there is likely to be
at least one that achieves 100\% heads. Among the 100 that we toss eight
times, there might be one that achieves 100\% heads, and there are
likely to be several that achieve 87.5\% heads. Among the 100 that we
toss ten times each, it is fairly unlikely than any will exceed
80\%, so the highest ranked units will almost certainly come from
among the coins that we toss only six times. That is, the highest
ranked units are almost all false-positives arising only out of
chance. This is still true, even if one or more of the coins that are
tossed ten times have a slightly higher probability of heads. 

On the other hand, if our main aim is to avoid false positives, we
could use a testing-based approach, where for each unit, we perform an
hypothesis test of whether the unit has some null status --- for
example whether the probability of heads is 0.5. We can then rank by
the $p$-values of these tests. This has the advantage of minimising
false positives, but in many cases there are a large number of true
positives, but only a few of them are truly important. If we apply the
testing approach, we will often select the units on which we have
collected most data, simply because the more data we have, the more
evidence that they are not null cases. This may lead to neglecting
some units which have much higher underlying value, but for which we
have less data.

Other approaches to the problem mainly take a Bayesian approach. They
assume that the true values of the relevant feature fall under some
distribution. We can estimate this underlying distribution from all
the data points. Then for each observed unit, we use this distribution
as a prior to estimate a posterior distribution of the true value for
this unit. We then perform our ranking based on these posterior
distributions and a choice of loss function. There are a range of
different methods based on different loss functions. For example,
posterior expected rank (Laird and Louis, 1989) use a loss function
the squared difference between the true rank of a unit (based on the
actual value of the feature) and the estimated rank. The $r$-values
method (Henderson and Newton, 2015), corresponds to a loss function
the sum of absolute differences between estimated rank and true
rank. Both of these loss functions are based entirely upon ranks, with
no consideration of the actual true values. That is, they consider
mis-ranking two units with almost identical true values to be as bad
as mis-ranking units with very different true values. For the vast
majority of practical ranking problems, this will not be the
case. Gelman and Price (1999) present the interesting case of looking
for spatial patterns among the top-ranked units, where artificial
patterns can arise from patterns in available sample sizes. For their
purposes, the ideal ranking method would be in such a way that the
distribution of rank is the same for all values of standard error. For
a known prior, it is possible to calculate this rank, though we are
not aware of any work applying such a ranking method. However, methods
with loss functions based only on rank, rather than value might be
expected to perform better on this criterion, since all errors in
ranking can cause this issue equally.

The aim of a ranking analysis is often to maximise the average true
value from the selected units. For instance, in the business profit
example, the aim would be to maximise the expected total profit. For
these purposes, the loss function is the difference between the
largest true values and the true values of selected units. This loss
function is introduced in Gupta and Hsiao (1983), with some
additional thought given to the situation where the loss is different
for the case of omitting a variable that should be included, from the
case of including a variable that should be omited. They show that for
this loss function with known prior the Bayes rule is to rank by
posterior mean (though they are not very explicit about this, and
include some unnecessary hypotheses). This posterior mean ranking is
used for example in Aitkin and Longford (1986). A range of other loss
functions have also been considered, for example, Lin {\em et al}.
(2006) summarise a range of choices of loss function. For this paper,
we will be focussing on the posterior mean ranking method, and its
corresponding loss function, although many of our methods can be
easily adapted to other Bayesian ranking methods.

The key difficulty in Bayesian ranking methods is to choose the form
of the prior. Two common choices are the conjugate prior (which for
normal error is normal), and a non-parametric prior, which can be
calculated using the results of
Laird~(1978). Figure~\ref{MotivatingProblem} shows the sort of problem
that can arise with this approach. The lines on that figure show
points that are ranked equally by posterior mean under a normal prior
estimated from the whole dataset. As can be seen in that plot, a lot
of emphasis gets placed on points with small variance. The reason is
that the normal prior is light-tailed, so large true values are deemed
implausible, and discounted. However, the true prior distribution
seems to be more heavy-tailed than the normal, so larger values should
not be discounted so much. For example, consider the point in the red
circle. While it does have a larger standard error, it is very
significantly non-zero, and it is likely that the true log-odds ratio
is high. Intuitively, we would probably want to rank this data point
among the very top-ranked units. However, the posterior mean under the
normal prior ranks it below a lot of other points which, while
certainly significantly non-zero, have very small effect size. For
practical purposes, this is not desirable. We are usually interested
in units with a large effect size.

\hfil\begin{figure}[htbp]
\hfil\includegraphics[width=10cm, clip=true,trim=0cm 1.5cm 0cm 1.5cm]{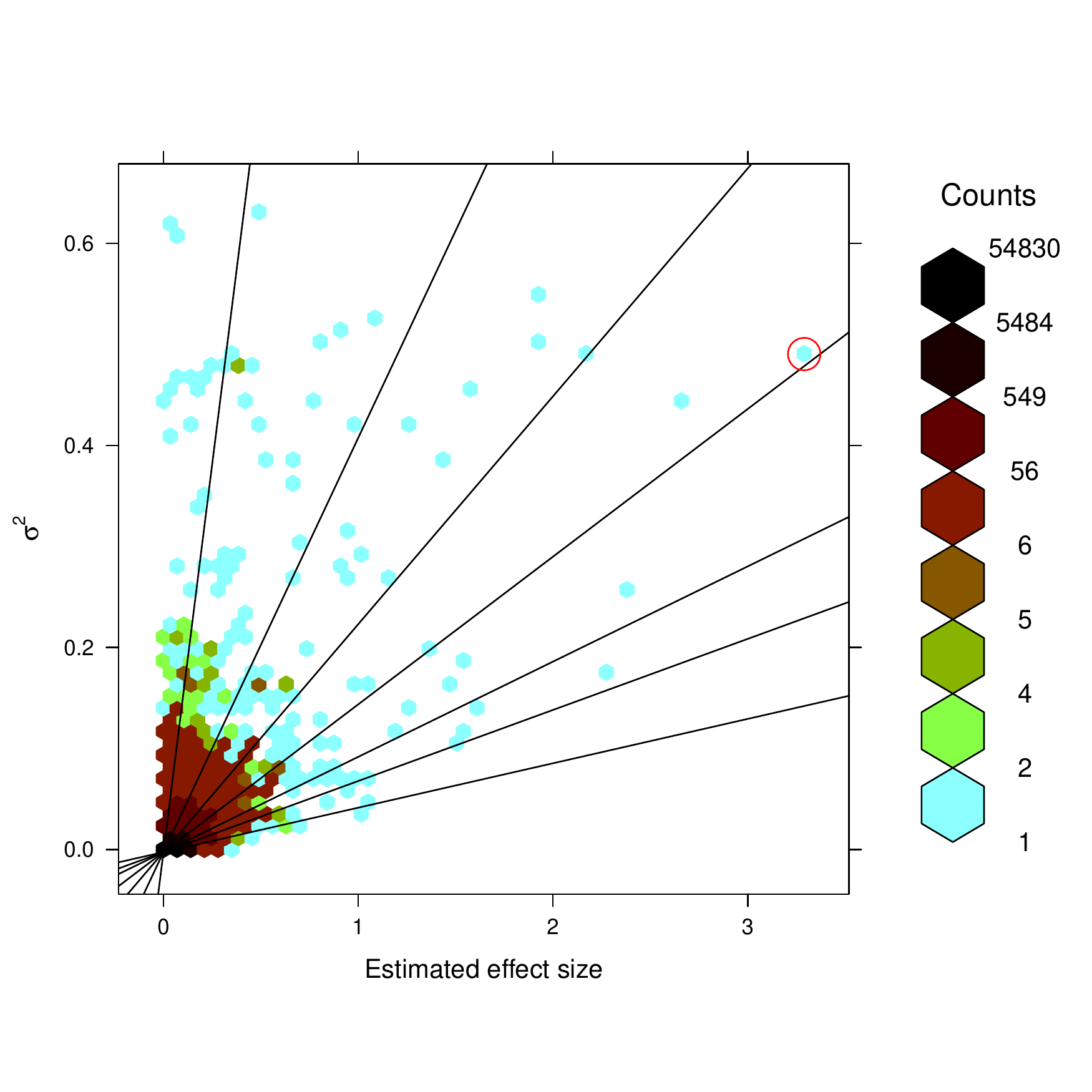}
\caption{Estimated log-odds ratio versus variance of estimator for SNP
data from a GWAS study into type-2 diabetes (Morris {\em et
  al.}, 2012). The lines show points ranked equally under posterior
mean with a normal prior.}\label{MotivatingProblem}
\end{figure}

The aim of this paper is to study the effect that choice of prior can
have on the ranking problem, and determine suitable choices of prior
for such analyses. Despite a fair amount of literature on Bayesian
ranking methods, there has been a noticeable lack of work on the
question of choice of prior. In view of the fact that selecting a
suitable model for the prior distribution is a very difficult problem
in model selection, it is important to consider the effects of a
misspecified prior distribution. As will become apparent later,
certain choices of prior are inherently more robust to
misspecification than others. Furthermore, some choices of prior are
more sensitive to parameter estimation than others. 

We describe the objective more formally as follows. A ranking problem
consists of a collection of units with unobserved parameters
$\theta_i$. For each unit, we have a point estimate $x_i$ for
$\theta_i$. We assume that $x_i$ is normally distributed with mean
$\theta_i$ and variance $\sigma_i{}^2$, where $\sigma_i$ is known. It
is straightforward to adapt our approach to a number of other error
distributions, but for this paper, we will focus on the normal error
case. We assume that the unobserved values $\theta_i$ follow what we
will refer to as the {\em true prior} distribution. We will rank by
posterior mean using what we will refer to as the {\em estimating
  prior}, which may or may not be the same as the true prior. We are
interested in how choice of the estimating prior affects the ranking.

The structure of this paper is as follows: In
Section~\ref{TheorySection}, we develop some theory behind posterior
mean ranking, and the loss from using the wrong estimating prior. In
Section~\ref{SectionIsotaxes}, we give a visual representation of the
effect of choice of estimating prior on posterior mean ranking. In
Section~\ref{SectionNonparametric}, we show that using the
non-parametric MLE as an estimating prior for posterior mean ranking
produces a robust ranking. In Section~\ref{Simulations}, we apply our
theory to some examples of misspecified estimating priors, and perform
a simulation study to confirm the results are as expected. We show
that an exponential estimating prior is a good general-purpose choice
for posterior mean ranking. In Section~\ref{SectionRealData}, we apply
this to some real data examples where we show the difference in the
ranking between using a normal distribution for the estimating prior
and using an exponential distribution. In
Section~\ref{SectionConclusions}, we make some concluding remarks and
suggestions for further investigations.

\section{Theory}\label{TheorySection}

\subsection{Approximate Posterior Mean for given Prior Distribution}

We suppose that our true prior distribution is continuous and has
density function $\pi(\theta)$. Suppose that we have a point estimate
$x$, whose error distribution is normal with variance $\sigma^2$,
where $\sigma$ is small. Since $\sigma$ is small, values of $\theta$
that are far from $x$ are extremely implausible, and contribute little
to the posterior mean for most choices of $\pi(\theta)$. We therefore
focus on the form of $\pi(\theta)$ for values of $\theta$ close to
$x$. Taking a first order Taylor expansion about $x$
gives $$\pi(\theta)=\pi(x)+\pi'(x)(\theta-x)$$ Using this
approximation to $\pi(\theta)$ gives that the posterior mean is
\begin{align*}
\frac{\int (x+(\theta-x))
  \left(\pi(x)+\pi'(x)(\theta-x)\right)e^{-\frac{(\theta-x)^2}{2\sigma^2}}\,d\theta}{\int \left(\pi(x)+\pi'(x)(\theta-x)\right)e^{-\frac{(\theta-x)^2}{2\sigma^2}}\,d\theta}&=x+\frac{\int (\theta-x) \left(\pi(x)+\pi'(x)(\theta-x)\right)e^{-\frac{(\theta-x)^2}{2\sigma^2}}\,d\theta}{\int
  \left(\pi(x)+\pi'(x)(\theta-x)\right)e^{-\frac{(\theta-x)^2}{2\sigma^2}}\,d\theta}\\
&=x+\frac{\pi(x)\int
  (\theta-x)e^{-\frac{(\theta-x)^2}{2\sigma^2}}\,d\theta+ \pi'(x)\int
    (\theta-x)^2 e^{-\frac{(\theta-x)^2}{2\sigma^2}}\,d\theta}{\pi(x)
    \int e^{-\frac{(\theta-x)^2}{2\sigma^2}}\,d\theta +\pi'(x)\int (\theta-x)e^{-\frac{(\theta-x)^2}{2\sigma^2}}\,d\theta}\\
&=x+\frac{\pi'(x)}{\pi(x)}\sigma^2\\
\end{align*}

This means that the key part of choice of estimating prior is to
estimate the quantity $\lambda(x)=-\frac{\pi'(x)}{\pi(x)}$. For the
tail of the distribution, this quantity is positive, and
asymptotically approaches the hazard rate. For an exponential
distribution, it is constant. For heavier-tailed distributions it
tends to zero as $x\rightarrow\infty$. For light-tailed distributions,
it tends to infinity as $x\rightarrow\infty$.

\subsection{Loss function in terms of posterior misestimation}

Suppose we should estimate the posterior mean as $x-\lambda\sigma^2$,
but in fact, we estimate it as $x-\hat{\lambda}\sigma^2$, for some
particular value of $x$. The question is what is the average loss
function resulting from this. 
For a ranking of all the observations, we can consider the total loss
as the sum of losses due to individual mis-rankings. That is, suppose we
rank the observations $\theta_{[1]},\theta_{[2]},\ldots,\theta_{[n]}$,
when the correct ranking is
$\theta_{(1)},\theta_{(2)},\ldots,\theta_{(n)}$. If we choose our
selection cutoff as the first $k$ units, then the loss function is
\begin{align*}
l_k&=\sum_{i=1}^k (\theta_{(i)}-\theta_{[i]})\\
&=\left(\sum_{\substack{i\leqslant k\\\theta_{(i)}\not\in\{\theta{[1]},\ldots,\theta{[k]}\}}}
\theta_{(i)}\right)-\left(\sum_{\substack{j\leqslant  k\\\theta_{[j]}\not\in\{\theta{(1)},\ldots,\theta{(k)}\}}}
\theta_{[j]}\right)\\
\end{align*}

We can move from the correct ranking to the estimated ranking by a
series of transpositions of adjacent units in the current ranking. For
example, if the correct ranking is $1,2,3,4,5,6$ and the estimated
ranking is $2,3,1,6,5,4$, we can change from the correct ranking to
the estimated ranking via the following sequence:
\begin{align*}
&1\ 2\ 3\ 4\ 5\ 6\\
&2\ 1\ 3\ 4\ 5\ 6\\
&2\ 3\ 1\ 4\ 5\ 6\\
&2\ 3\ 1\ 4\ 6\ 5\\
&2\ 3\ 1\ 6\ 4\ 5\\
&2\ 3\ 1\ 6\ 5\ 4\\
\end{align*}

For each such transposition, exchanging the position of $\theta_{(i)}$
in the $m$th postion, with $\theta_{(j)}$ in the $(m+1)$th position,
the change in loss function is
$$\left\{\begin{array}{ll}\theta_{(i)}-\theta_{(j)}&\textrm{if }m=k\\0&\textrm{otherwise}\end{array}\right.$$
The total loss from
this mis-ranking is then given by the sum of the loss functions for
each transposition. We see that the loss for each transposition is
non-negative for each value of $k$, so we can analyse the overall loss of
a misranking by looking at the loss of each pairwise misranking. 

If we consider the overall loss as the total of the loss functions for
all values of $k$, we see that this loss function is just the sum of
the loss functions for each transposition. Furthermore, whatever
sequence of transpositions is performed, there will be one
transposition for each misranked pair. Therefore the total loss
function is the sum of the losses from each misranked pair. We can
therefore study the total loss function by studying the misranking
loss for any pair of observations. In practice, we will often consider
only the loss of the upper tail of the distribution. That is, we will
choose some cutoff $a$ and evaluate the sum of the loss function for
all $k$ such that $x_{(k)}>a$. For this we have the following
proposition (proof in Appendix~\ref{AppProofOfLossFunction})

\begin{proposition}
Suppose the true prior distribution of the parameter $\theta$ has
density function $\pi(\theta)$, and that we have two observations
$x_1$ and $x_2$ which are normally distributed with means $\theta_1$
and $\theta_2$ and standard deviations $\sigma_1$ and $\sigma_2$
respectively, where $\theta_1$ and $\theta_2$ are random samples from
the true prior distribution, and $\sigma_1$ and $\sigma_2$ are assumed
to be small.

\begin{enumerate}
\renewcommand\labelenumi{(\roman{enumi})}
\renewcommand\theenumi\labelenumi

\item The expected loss when the estimating prior and the true prior
  are the same (which we will refer to as the {\em optimal expected
    loss}) is approximately given by
$$\frac{\sigma_1{}^2+\sigma_2{}^2}{2}{\mathbb
  E}(\pi(x))$$

\item When the estimating prior has density $\hat{\pi}$, the
  difference between the expected loss and the optimal expected loss
  is approximately given by
\begin{equation}\label{approxtotalloss}\frac{1}{2}(\sigma_1{}^2-\sigma_2{}^2)^2\int_{a}^\infty \pi(x)^2(\lambda(x)-\hat{\lambda}(x))^2\,dx\end{equation}
where $\hat{\lambda}(x)=-\frac{\hat{\pi}'(x)}{\hat{\pi}(x)}$.

\item The difference between the expected loss from using the point
  estimate $x_i$ and the optimal expected loss is approximately given by

$$\frac{1}{2}\left(\sigma_1{}^2-\sigma_2{}^2\right)^2\int_a^\infty \pi'(x)^2\,dx$$

\end{enumerate}
\end{proposition}

We see that for $\hat\lambda(x)\leqslant\lambda(x)$, $\int_{a}^\infty
\pi(x)^2(\lambda(x)-\hat{\lambda}(x))^2\,dx$ is bounded by
$\int_{-\infty}^\infty \pi(x)^2\lambda(x)^2\,dx$, which is the
expected information of $\theta$, and is bounded for most
distributions. This means that if the estimating prior is too
heavy-tailed, we can do no worse than ranking by point estimators
alone. On the other hand, if we have $\hat{\lambda}(x)\geqslant
\lambda(x)$, then the integral can approach $\int_{-\infty}^\infty
\pi(x)^2\hat{\lambda}(x)^2\,dx$, which can be unbounded if the true
prior has a heavy tail, but the estimating prior has a light tail. In
most cases, the expression will not be unbounded. For example, for a
normal estimating prior and a Pareto true prior, we have that
$\hat{\lambda}(x)=\frac{x}{\tau^2}$ and
$\pi(x)=\frac{\alpha\eta^\alpha}{x^{\alpha+1}}$, so
$$\int_0^\infty
\pi(x)^2\hat{\lambda}(x)^2\,dx=\frac{\alpha^2\eta^{2\alpha}}{\tau^4}\int_0^\infty
x^{2-2(\alpha+1)}\,dx=\frac{\alpha^2\eta^{2\alpha}}{\tau^4}\int_0^\infty
x^{-2\alpha}\,dx$$ which diverges whenever
$\alpha\leqslant\frac{1}{2}$. Thus for very heavy-tailed true priors,
the loss from using a light-tailed estimating prior can diverge.

We see that there is a risk of this unbounded loss whenever
$\hat{\lambda}(x)$ diverges. This can happen for any estimating prior
with a lighter tail than an exponential distribution. We therefore
suggest using an exponential distribution for the estimating prior to
ensure the loss is not too great.  This has the added mathematical
convenience that the posterior mean is easily calculated as
$\hat{\theta}=x-\lambda\sigma^2$ for some constant $\lambda$. If we
use an improper exponential prior with density proportional to
$e^{-\lambda \theta}$ for all $\theta$ (not just $\theta>0$) then this
formula for the posterior mean is exact. Indeed the posterior
distribution is given by
\begin{align*}
\pi_x(\theta)&\propto
e^{-\lambda\theta}e^{-\frac{(\theta-x)^2}{2\sigma^2}}=e^{-\frac{(\theta-x+\lambda\sigma^2)^2}{2\sigma^2}+\frac{\lambda^2\sigma^2}{2}-\lambda  x}\propto e^{-\frac{(\theta-x+\lambda\sigma^2)^2}{2\sigma^2}}\\
\end{align*}
which is the density of a normal distribution with mean
$x-\lambda\sigma^2$ and variance $\sigma^2$.

In this proposition, part~(ii) gives the measure of the cost of
using the wrong estimating prior. (i) and (iii) give measures of the overall
difficulty of the ranking problem. (i) is the irreducible cost of
misranking. (iii) is the additional cost from using the point
estimates to rank, instead of using the posterior mean. It is an
indication of the extent to which the ranking can be improved by using
Bayesian methods.

\section{Shapes of ranking thresholds}\label{SectionIsotaxes}

Henderson and Newton (2015) describe different ranking methods in
terms of the shapes of what they refer to as ``threshold functions'',
namely the functions $t_\alpha(\sigma^2)$ which are the smallest value
of $x$, such that the observation $(x,\sigma^2)$ is ranked in the top
$\alpha$ proportion under the ranking method in question. These
threshold functions are curves joining points of equal rank: we will
therefore refer these curves as {\em isotaxes} (singular: {\em
  isotaxis}, from Greek {\em iso} meaning equal, and {\em taxis}
meaning rank). Henderson and Newton (2015) then describe their
$r$-values procedure directly by calculating the shape of these
isotaxes. We will examine the shape of the isotaxes as a method to
better determine the effect of the estimating prior on ranking.

For Bayesian methods, the shape of these isotaxes depends heavily on
the choice of estimating prior. For the normal estimating prior with
mean 0 and variance $\tau^2$, for an observation $x$ with standard
error $\sigma$, the posterior mean is
$\frac{\tau^2}{\tau^2+\sigma^2}x$, so isotaxes are given by solutions
to $\frac{\tau^2}{\tau^2+\sigma^2}x=C$ for constant $C$, or to
$\sigma^2=\frac{\tau^2}{C}x-\tau^2$. When plotted on a graph of
$\sigma^2$ against $x$, these are lines of varying slope, with
shallower slope at higher ranks. (Indeed, these lines all pass through
the point $(0,-\tau^2)$.)

For an exponential estimating prior with hazard rate $\lambda$, as
mentioned above, the posterior mean is given by $x-\lambda\sigma^2$.
The isotaxes are therefore given by the equation
$x-\lambda\sigma^2=C$, or
$\sigma^2=\frac{x}{\lambda}-\frac{C}{\lambda}$, so they are lines of
constant slope.

For a heavy-tailed distribution, recall that we have posterior mean
approximately $x+\frac{\pi'}{\pi}\sigma^2$. Therefore 
the isotaxes
are functions of the form $x+\frac{\pi'(x)}{\pi(x)}\sigma^2=C$. A typical
example is $\pi(x)=x^{-\alpha}$, so that
$\frac{\pi'(x)}{\pi(x)}=-\frac{\alpha}{x}$. This means the isotaxes
are curves of the form
\begin{align*}
x-\frac{\alpha\sigma^2}{x}&=C\\
x^2-\alpha\sigma^2&=Cx\\
\sigma^2&=\frac{1}{\alpha}\left(x-\frac{C}{2}\right)^2-\frac{C^2}{4\alpha}\\
\end{align*}

which gives a parabola. We plot the shapes of the isotaxes for these
estimating prior distributions in Figure~\ref{Fig_isotaxis}.

\begin{figure}[htbp]
\caption{Isotaxis plots for various choices of estimating prior distribution
  using posterior mean ranking}\label{Fig_isotaxis}
\begin{subfigure}{0.33\textwidth}\caption{Normal}
\includegraphics[width=5.5cm]{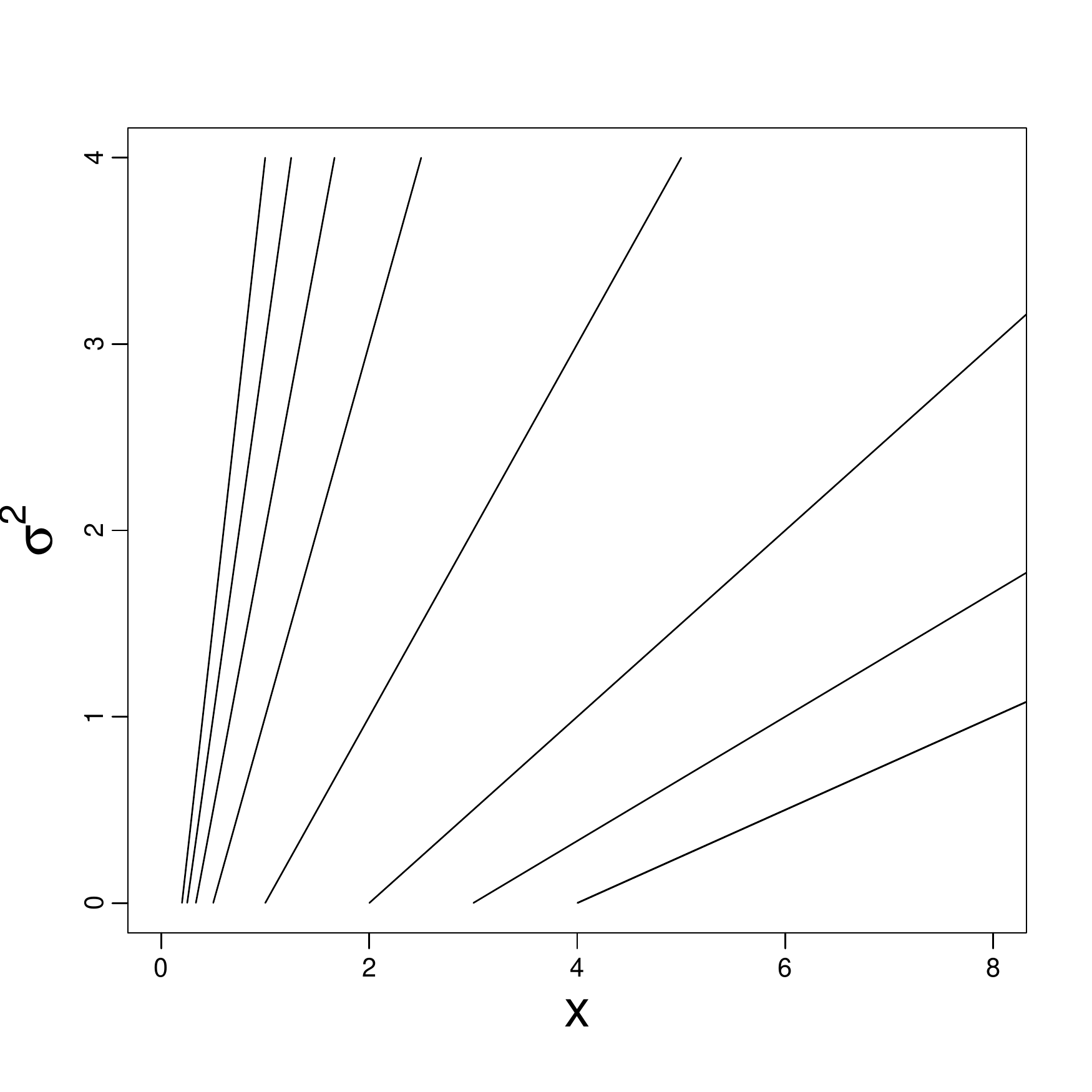}
\end{subfigure}
\begin{subfigure}{0.33\textwidth}\caption{Exponential}
\includegraphics[width=5.5cm]{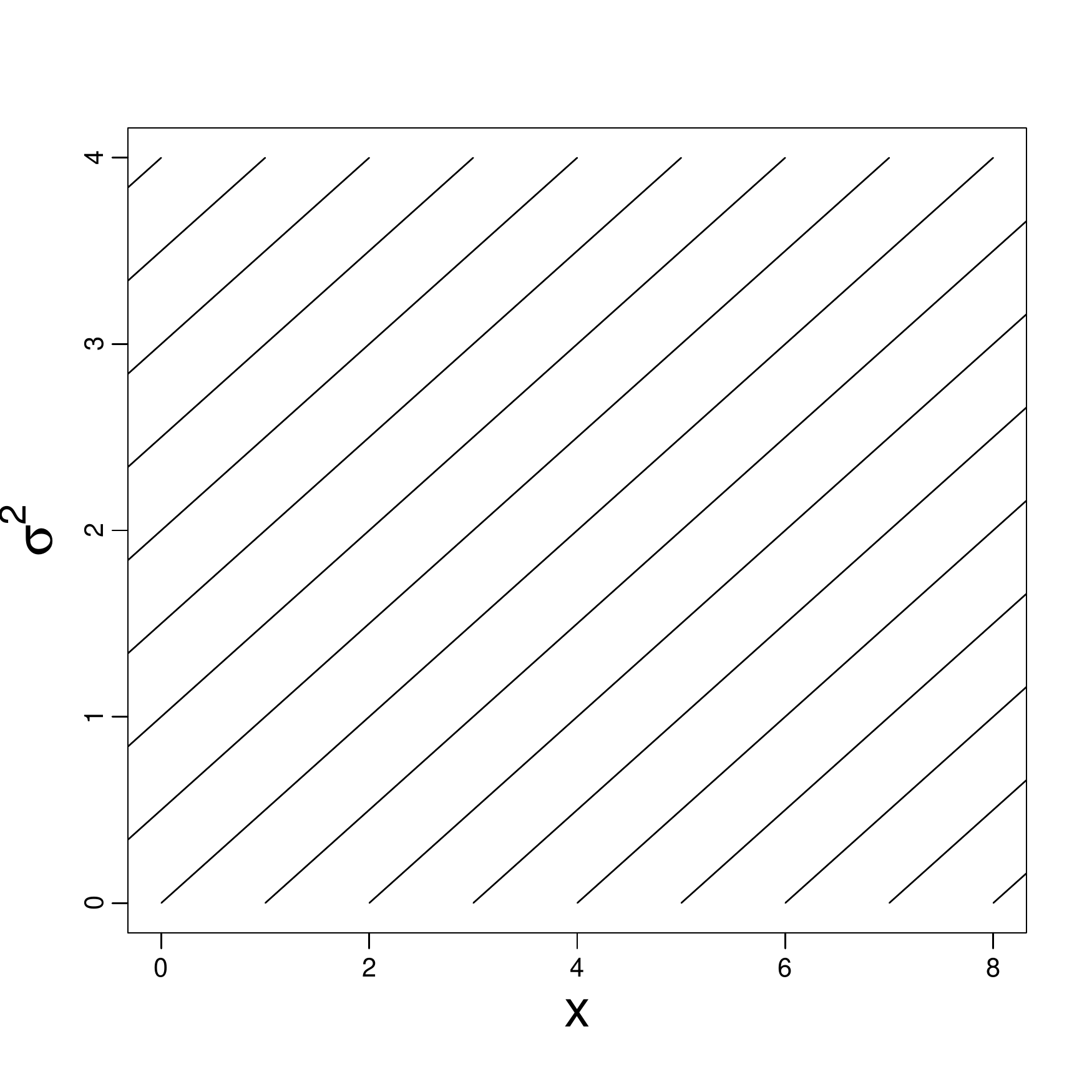}
\end{subfigure}
\begin{subfigure}{0.33\textwidth}\caption{Pareto}
\includegraphics[width=5.5cm]{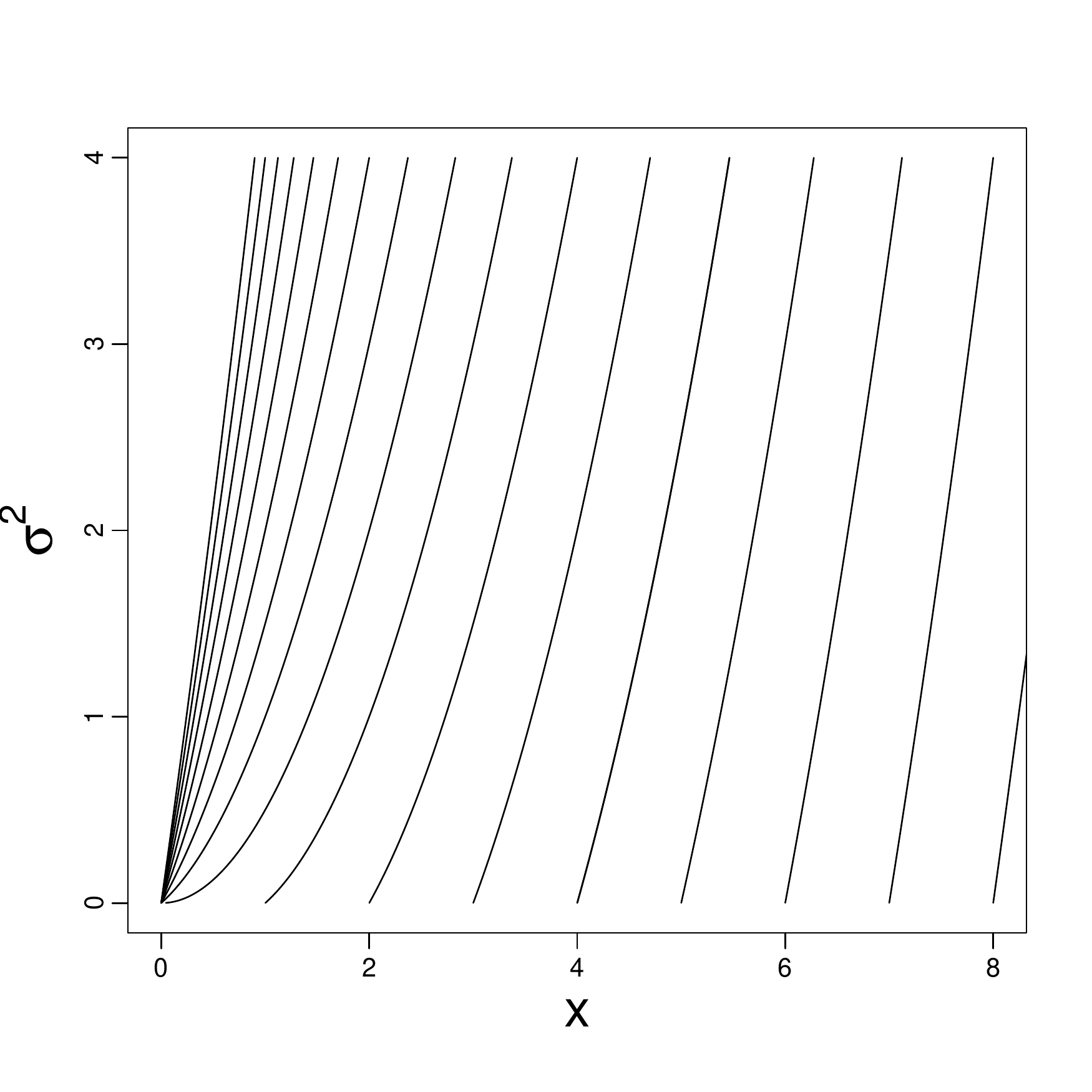}
\end{subfigure}

\end{figure}

We see that for the exponential and heavy-tailed estimating priors,
the slopes of isotaxes are bounded away from zero, so the posterior
mean cannot be very far from the point estimate for $x$. Since by
assumption, the true value also will not be so far from the point
estimate, this means that the posterior mean cannot be too far from
the true value.

From the shapes in Figure~\ref{Fig_isotaxis}, we see that for the
normal estimating prior, the standard error becomes increasingly
important as we move towards the tail of the distribution, and that
the posterior mean can be arbitrarily far away from the true
value. For the exponential estimating prior, the standard error
remains equally important throughout. For the heavy-tailed estimating
prior, the standard error becomes less important as we move to the
tail of the distribution. Furthermore, the standard error is most
important for small standard error, and differences in standard error
become less important as the standard error increases.

\section{Non-parametric Prior}\label{SectionNonparametric}

It is also possible to calculate a non-parametric maximum likelihood
estimate for the prior distribution. It was shown by Laird (1978) that
the prior in this case is a discrete distribution with finite
support. An implementation of this non-parametric prior estimation is
given in the \texttt{rvalues} package in \texttt{R}. However, this
implementation is buggy, so we were unable to compare this method in
Section~\ref{Simulations}. We show that for such a choice of
estimating prior, provided the support of the prior distribution
includes points sufficiently close to all the observed data, then the
posterior mean estimators are robust.  Proofs of the following lemmas
are in Appendix~\ref{NPPriorProof}.

\begin{lemma}
Let $\pi$ be a discrete distribution with probability at least
$\frac{1}{r+1}$ in the interval $[x-a,x+a]$ for some $a>0$. Let
$\hat{\theta}$ be the posterior mean for an observation $x$ with
standard error $\sigma$. Then
$$|\hat{\theta}-x|\leqslant a+\sigma\sqrt{2\log(r)}$$
\end{lemma}

This means that provided the prior distribution assigns some
probability to a region near to each observed value of $x$, then
the posterior mean estimate will have some robustness to model
misspecification.

\begin{lemma}
For a sample of $n$ datapoints and their corresponding standard
errors, the non-parametric MLE estimate for the prior distribution
always assigns probability at least $\frac{1-e^{-\frac{1}{2}}}{n}$ to
the interval $(x-\sigma\sqrt{2\log(n)+1},x+\sigma\sqrt{2\log(n)+1})$,
for every observed data point $(x,\sigma)$.
\end{lemma}

From the preceding lemmas, we conclude that ranking based on posterior
mean under the non-parametric MLE estimate for the prior is relatively
robust, with 
$$|\hat{\theta}-x|\leqslant
\sigma\left(\sqrt{2\log(n)+1}+\sqrt{2\log\left(\frac{n}{\left(1-e^{-\frac{1}{2}}\right)}-1\right)}\right)$$
We also know that for large $n$, the non-parametric MLE estimate is
consistent, so the ranking will be optimal with the non-parametric
MLE. Overall, we conclude that non-parametric estimation of prior
provides a reasonable compromise between efficiency and robustness.

However, as is typically the case with non-parametric methods, there
is a trade-off between bias and variance. For the non-parametric
method, the estimated ranking is asymptotically unbiassed, but can
have fairly large variance for smaller sample
sizes. Figure~\ref{NPsimulationExample} gives an illustration of this.

\begin{figure}

\caption{Comparison of Isotaxes for Non-parametric and Parametric Estimation}\label{NPsimulationExample}

\begin{subfigure}{0.33\textwidth}\caption{Non-parametric}
\includegraphics[width=5.5cm]{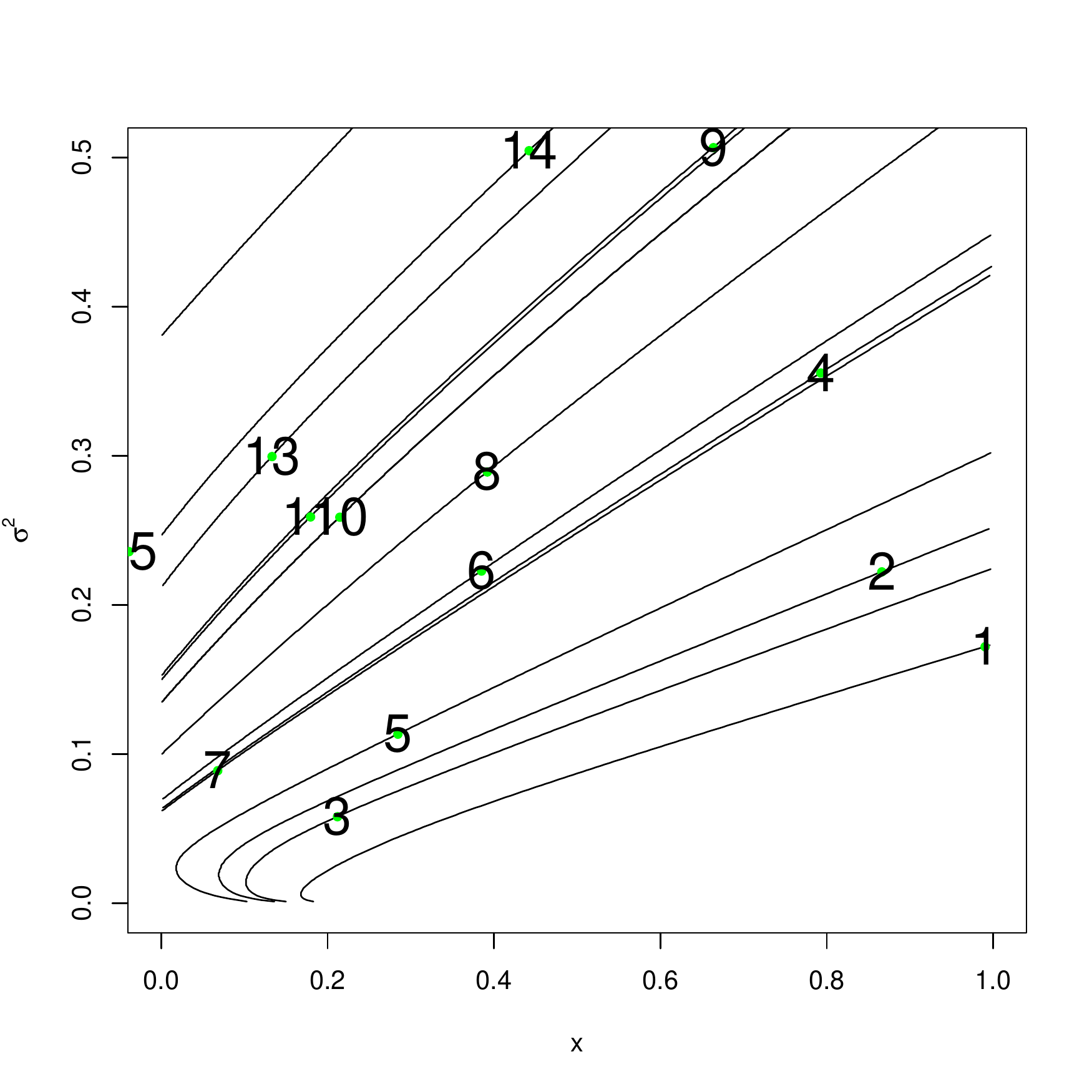}
\end{subfigure}
\begin{subfigure}{0.33\textwidth}\caption{Exponential Estimating Prior}
\includegraphics[width=5.5cm]{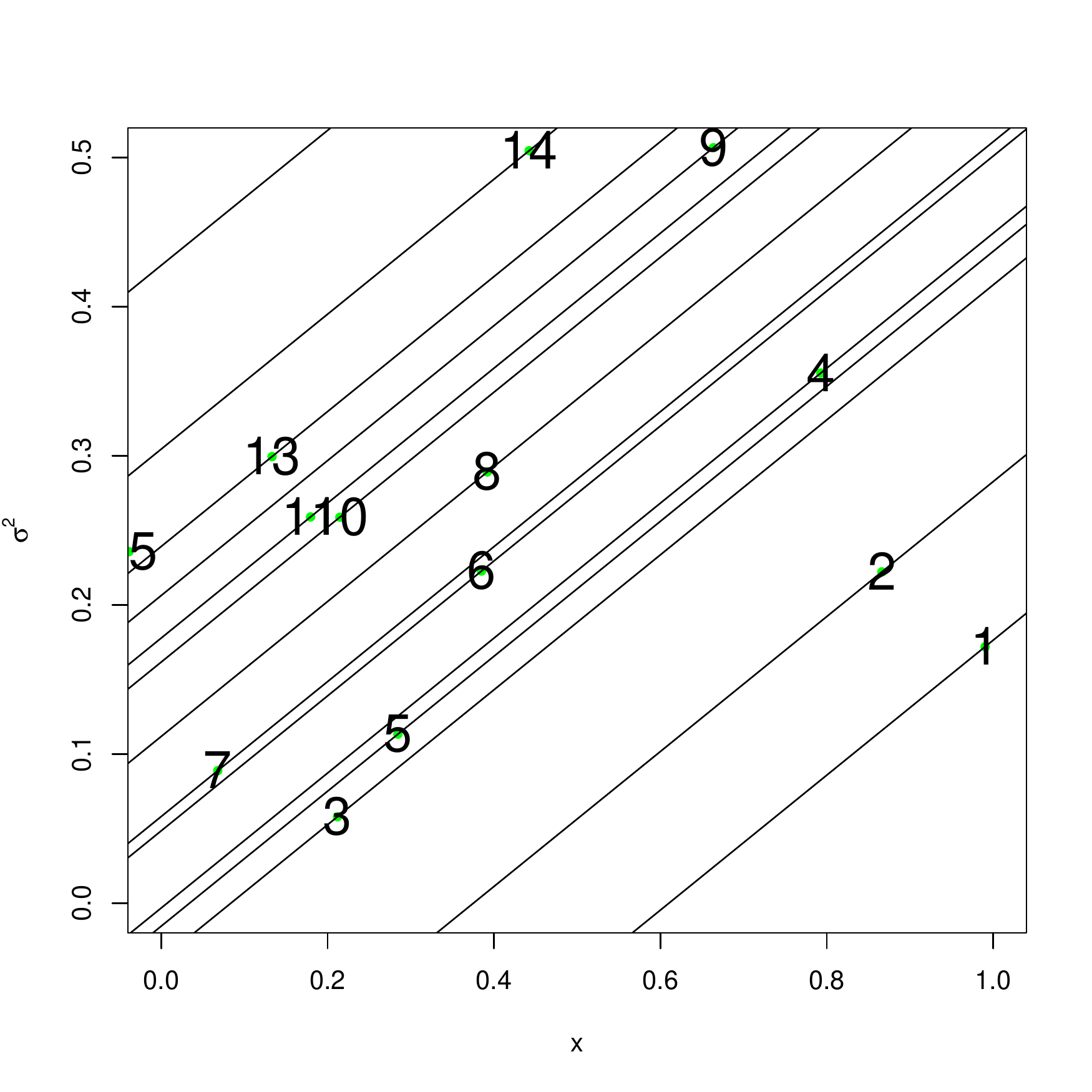}
\end{subfigure}
\begin{subfigure}{0.33\textwidth}\caption{True Prior}
\includegraphics[width=5.5cm]{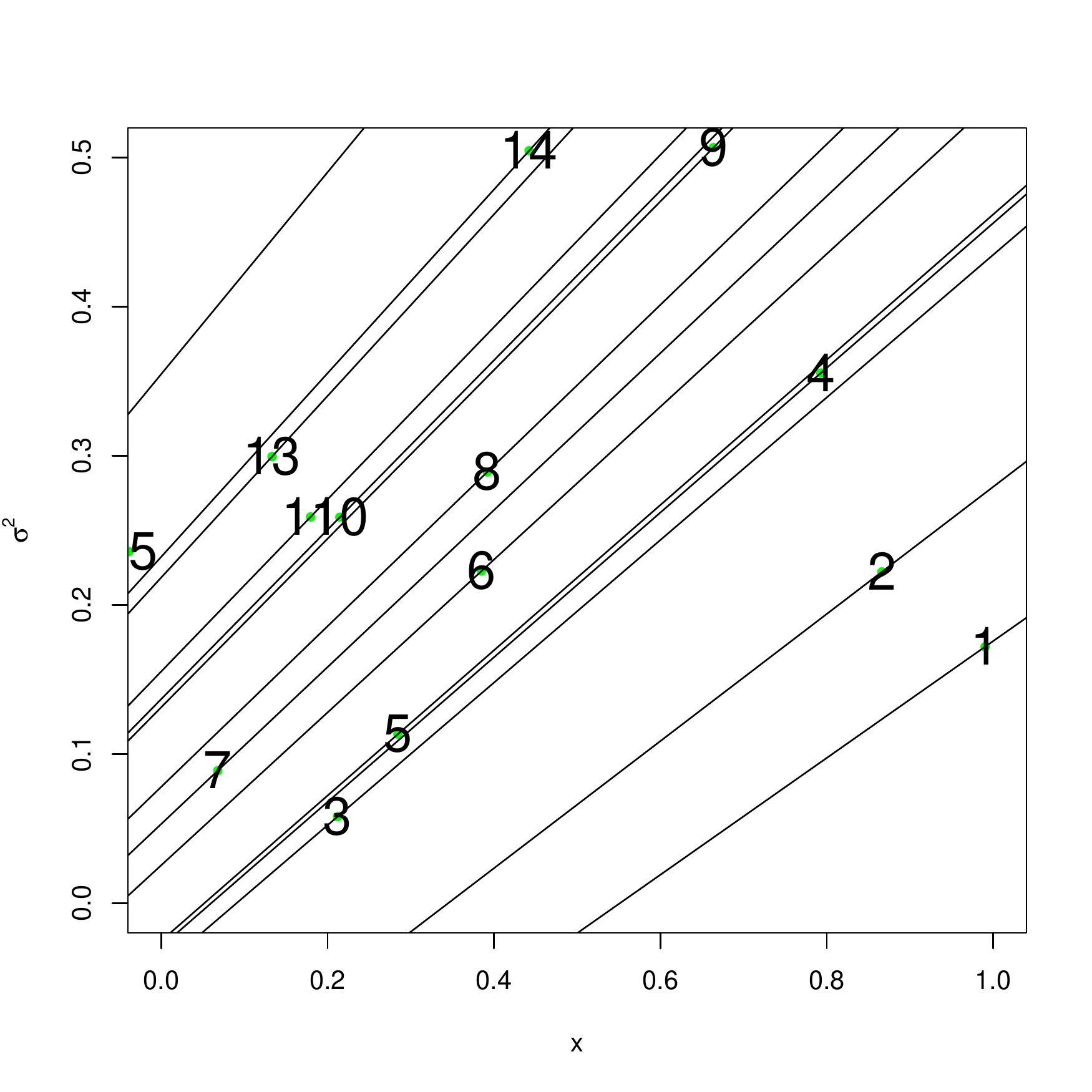}
\end{subfigure}

Isotaxes for the upper tail of simulated data. 500 data points were
simulated with the true means following a normal distribution with
mean $-2.3$ and variance 1. Variances for the observed data points are
simulated following a gamma distribution with shape parameter 2 and
scale parameter 0.1. Plot (a) shows the isotaxes for the
non-parametric MLE estimate for the prior distribution. Plot (b) shows
the isotaxes for an exponential estimating prior. Plot (c) shows the
isotaxes for the true prior. Points are numbered according to their
rank by posterior means under the true prior. Note that some points
are outside the region shown, hence the missing numbers. The isotaxes
shown are the ones passing through observed data points.

\end{figure}

We can see that while the non-parametric approach has the isotaxes in
approximately the right direction for larger variances, they are
somewhat distorted for smaller variances. This is particularly
observable at the tail, because the support of the MLE (which is
discrete by Laird~(1978)) is fairly sparse around the tail. This has a
big effect on the posterior mean estimates for points with small
standard error. However, it is worth noting that this distortion
usually has limited influence on the estimated ranking. The reason for
this is that the distortion is only for small standard error, compared
with the standard error of the data points, so if some of the data
points have small standard error, the isotaxes for posterior mean
ranking based on the MLE prior will be close to the correct isotaxes
except for very small standard error. Meanwhile, if the standard
errors are large, the MLE isotaxes will become further from the
correct isotaxes, but not many of the observed data points will be
included in this region where the isotaxes are far from optimal. The
example given in Figure~\ref{NPsimulationExample} is a typical example
where the non-parametric MLE prior gives a poor ranking. There are
other typical examples where the MLE prior does not give such a poor
ranking.

\section{Simulation}
\label{Simulations}

\subsection{Simulation Design}

We use three simulation distributions for the priors (both the true
priors and the estimating priors): A normal distribution with known
mean 0 and variance $\tau^2$; An exponential distribution with hazard
rate $\lambda$; and a Pareto distribution with density function
$\pi(\theta)=\frac{\alpha\eta^\alpha}{\theta^{\alpha+1}}$ for
$\theta>\eta$ where we take $\eta=\frac{1}{2}$ as known. (We have
taken one parameter as known for the normal and Pareto distributions,
so that each prior has one hyperparameter to be estimated.) For the
true priors in the simulation, we set $\tau=1$ for the normal
distribution, $\lambda=1$ for the exponential distribution and
$\alpha=2$ for the Pareto distribution. For each simulation
distribution, we simulate datasets of size 1000, 10000, and 100000.

We simulate the standard error $\sigma$ for each data set as following
an exponential distribution. We present results for the mean of this
exponential distribution equal to 0.02. Results for mean 0.01, 0.05
and 0.1 are presented in the supplementary materials. The values of
$\sigma$ are independent of the values of $\theta$ and values of
$\sigma$ for different data points are independent. To avoid some
computational issues caused by values of $\sigma$ too close to 0, we
added 0.0001 to all values of $\sigma$. We do not expect this to
significantly impact the results, but we found that some numerical
integration routines produced errors when the value of $\sigma$ was
very close to zero.

For each simulated dataset, we analyse with each of the normal,
exponential and Pareto distributions as the estimating prior. We will
assess the performance of the ranking by the average increase in the
loss function from using the given estimating prior compared to using
the true prior. That is, the loss function is:
\begin{equation}\label{LossFunctionEqn}
L=\sum_{i=1}^{0.1k} (0.1k-i)(\theta_{(i)}-\theta_{[i]})
\end{equation}
where $\theta_{(i)}$ is the true value of $\theta$ for the $i$th
ranked unit under the true prior, and $\theta_{[i]}$ is the true value
of $\theta$ for the $i$th ranked unit under the estimating prior.

\subsection{Theoretical Analysis of Expected Loss for Simulation Distributions}

In order to better understand issues related to parameter estimation,
we examine the loss function for both optimal parameter estimates
(based on minimising the expected loss function) and estimated
parameters (estimated from the upper tail of the data). We do not
compare the effect of estimating the hyperparameters from the whole
data set because two of the distributions used for analysis had
support only on the positive real numbers, so estimating these based
on the whole data including negative values might lead to strange
results. Even if the supports were all the same, estimating the
parameters for the estimating prior based on the whole data set when
the focus is on the ranking of the top units leads to suboptimal
results in ranking.

\begin{table}
\renewcommand{\arraystretch}{2.5}
\caption{Expected Loss functions}\label{LossFunctions}
\begin{tabular}{lll}
\hline
True & Estimate & Loss function\\
\hline
Normal & Normal & $\displaystyle \left(\frac{1}{\tau^2}-\frac{1}{\hat{\tau}^2}\right)^2\left(\frac{a
e^{-\frac{a^2}{\tau^2}} }{4\pi}+\frac{\tau}{4\sqrt{\pi}}\left(1-\Phi\left(\frac{\sqrt{2}a}{\tau}\right)\right)\right)$\\
Normal & Exponential & $\displaystyle \frac{\hat{\lambda}^2}{2\sqrt{\pi}\tau}\left(1-\Phi\left(\frac{\sqrt{2}a}{\tau}\right)\right)
-\frac{\hat{\lambda}e^{-\frac{a^2}{\tau^2}}}{2\pi\tau^2}+\frac{ae^{-\frac{a^2}{\tau^2}}}{4\pi\tau^4}+\frac{1-\Phi\left(\frac{\sqrt{2}a}{\tau}\right)}{4\sqrt{\pi}\tau^3}$ \\
Normal & Pareto & $\displaystyle\frac{1}{2\pi\tau^2}\left((\hat{\alpha}+1)^2\int_a^\infty
\frac{1}{\theta^2}e^{-\frac{\theta^2}{\tau^2}}\,d\theta-\left(2\hat{\alpha}+\frac{3}{2}\right)
\frac{\sqrt{\pi}}{\tau}\left(1-\Phi\left(\frac{\sqrt{2}a}{\tau}\right)\right) +\frac{ae^{-\frac{a^2}{\tau^2}}}{2\tau^2}\right)$\\
Exponential & Normal & $\displaystyle\frac{e^{-2\lambda a}}{\lambda}\left(\frac{\lambda^4}{2}-\frac{2\lambda^3
  a+\lambda^2}{2\hat{\tau}^2}+\frac{2\lambda^2a^2+2\lambda
  a+1}{4\hat{\tau}^4}\right)$ \\
Exponential & Exponential & $\displaystyle \frac{\left(\lambda-\hat{\lambda}\right)^2}{2\lambda}e^{-2\lambda a}$\\
Exponential & Pareto & $\displaystyle\lambda^2\left((\hat{\alpha}+1)^2\int_a^\infty
\frac{1}{\theta^2}e^{-2\lambda\theta}\,d\theta-2\lambda(\hat{\alpha}+1)\int_a^\infty
\frac{1}{\theta}e^{-2\lambda\theta}\,d\theta +\frac{\lambda}{2}
e^{-2\lambda a}\right)$\\
Pareto & Normal
&
$\displaystyle\frac{\alpha^2\eta^{2\alpha}}{a^{2\alpha+3}}\left(\frac{(\alpha+1)^2}{(2\alpha+3)}-\frac{2(\alpha+1)a^2}{(2\alpha+1)\hat{\tau}^2}+\frac{a^4}{(2\alpha-1)\hat{\tau}^4}\right)$\\
Pareto & Exponential &
$\displaystyle\frac{\alpha^2\eta^{2\alpha}}{a^{2\alpha+3}}\left(\frac{(\alpha+1)^2}{(2\alpha+3)}-\hat{\lambda}a+\frac{\hat{\lambda}^2a^2}{(2\alpha+1)}\right)$\\
Pareto &  Pareto & $\displaystyle \frac{\alpha^2\left(\alpha-\hat{\alpha}\right)^2\eta^{2\alpha}}{(2\alpha+3)a^{2\alpha+3}}$\\
\hline
\end{tabular}
\renewcommand{\arraystretch}{1}
\end{table}

We calculate the expected loss function in each case (details in
Appendix~\ref{OptimalParameters}). Table~\ref{LossFunctions} gives the
expected loss function (using Equation~(\ref{approxtotalloss})) as a
function of the true and estimated parameters for each scenario.  The
optimal parameter values for the estimating priors are therefore the
values that minimise these loss functions. Table~\ref{SimExpLoss}
gives the optimal parameter values in all scenarios, and the
corresponding expected additional loss in each scenario from using the
misspecified estimating prior distribution. The final column uses the
point estimate instead of posterior mean ranking.

\begin{table}[htbp]

\caption{Expected loss --- Optimal parameter values (values which
  minimise the expected loss function from Table~\ref{LossFunctions})
  and the resulting values of the expected loss
  function compared with the true prior}\label{SimExpLoss}

\begin{subtable}{0.48\textwidth}
\caption{Optimal Parameter Values}
\begin{tabular}{ll|lll}
\hline
&&\multicolumn{3}{c}{Estimating Prior}\\
\cline{3-5}
&&Normal & Exp. & Pareto \\
&&$\hat{\tau}$ & $\hat{\lambda}$ & $\hat{\alpha}$\\
\hline
\multicolumn{1}{c|}{\multirow{3}{2em}{True Prior}} & Normal ($\tau=1$) & 1 &
1.561 & 1.290\\
\multicolumn{1}{c|}{} & Exp.  ($\lambda=1$) & 1.701 & 1 &  1.677 \\
\multicolumn{1}{c|}{} & Pareto  ($\alpha=2$)& 1.179 & 1.581 & 2 \\
\hline
\end{tabular}
\end{subtable}
\begin{subtable}{0.53\textwidth}
\caption{Expected Loss}
\begin{tabular}{ll|llll}
\hline
&&\multicolumn{3}{c}{Estimating Prior}& Point\\
\cline{3-5}
&&Normal & Exp. & Pareto & Estimate\\
\hline
\multicolumn{1}{c|}{\multirow{3}{2em}{True Prior}} & Normal ($\tau=1$)& 0 &
0.00062 & 0.00208 & 0.0247 \\
\multicolumn{1}{c|}{} & Exp.  ($\lambda=1$)& 0.00015 & 0 & 0.00010 & 0.005 \\
\multicolumn{1}{c|}{} & Pareto  ($\alpha=2$)& 0.00208 & 0.00036 & 0 &
0.0130 \\
\hline
\end{tabular}
\end{subtable}
\end{table}

From Table~\ref{LossFunctions}, We see that the loss functions are
quadratic in the parameters of the estimating prior (or in
$\frac{1}{\tau^2}$ for the normal distribution). This means that the
sensitivity of the loss function to misestimation of the parameter
values is roughly proportional to the mean squared difference between
the parameter estimate and the optimal value. We calculate the
constants of proportionality for our particular choices of parameter
values in Table~\ref{MisestimateLoss}. This gives a measure of the
sensitivity of the loss function to errors in parameter estimation.

\begin{table}[htbp]

\caption{Misestimation Loss --- The loss functions are all quadratic
  in the estimated parameter (or $\frac{1}{\tau^2}$ for the
  normal). This table gives the second derivative of the loss function
  --- it gives an indication of the relative cost of misestimating the
  parameter values. Optimal parameter estimates are given in
  Table~\ref{SimExpLoss}(a). If the optimal parameter value is
  $\theta$ and the value used is $\theta'$, then the additional loss
  is $a(\theta-\theta')^2$, where $a$ is the number in this table. For
  the normal estimating prior, the parameter to be estimated is
  $\frac{1}{\tau^2}$, rather than $\tau$. For the exponential it is
  the rate $\lambda$. For the Pareto, it is the index $\alpha$. This
  table gives a measure of the sensitivity of each estimating prior to
  misestimation of the parameter.  }\label{MisestimateLoss}

\hfil\begin{tabular}{ll|lll}
\hline
&&\multicolumn{3}{c}{Estimating Prior}\\
\cline{3-5}
&&Normal & Exp. & Pareto \\
\hline
\multicolumn{1}{c|}{\multirow{3}{2em}{True Prior}} 
                      & Normal      & 0.04933429 & 0.009862926 & 0.004307 \\
\multicolumn{1}{c|}{} & Exponential & 0.04052242 & 0.005 & 0.0006835 \\
\multicolumn{1}{c|}{} & Pareto      & 0.02108185 & 0.005059644 & 0.001445613 \\
\hline
\end{tabular}

\end{table}

As we see in Table~\ref{MisestimateLoss}, the normal estimating prior
is most sensitive to parameter estimation. This makes sense, since the
variance of the normal distribution has a very significant impact on
the slopes of the isotaxes in the tail of the distribution. The
exponential estimating prior is less sensitive to misestimation of
parameter values, and the Pareto estimating prior is least sensitive to
parameter estimates. This is because in the tail of the distribution,
the isotaxes for the Pareto estimating prior become very steep, regardless
of the parameter estimates. This indicates an advantage of using a
heavy-tailed estimating prior, particularly for small sample sizes,
where our parameter estimates have higher MSE. Even for large sample
sizes, the parameter estimates are likely to be different from the
optimal values, because we typically estimate parameters by a method
such as MLE, based on the observed data.  We were only able to
optimise the loss function for the simulations where we knew the true
prior distribution, but in a real situation we would not know the true
prior. The parameters estimated by MLE are not optimal for posterior
mean ranking.

We now look at the question of parameter estimation. Because we are
interested in fitting the tail of the distribution well, we truncate
the distribution at the 90th percentile (for the simulations, we used
the 90th percentile of the true prior), and estimate the parameters by
maximum likelihood for the truncated distribution. Details of the MLE
estimates, with derivation, are in Appendix~\ref{AppMLEestderivation}
We compare the theoretically best values and the expected MLE
estimates in Table~\ref{TheoEstPar}. (The Pareto distribution used for
simulation has infinite variance, so the MLE estimate for the normal
variance does not converge to a constant as sample size increases.)
Some of the MLE estimates used here are approximate, so may not
exactly reflect parameter values; empirical mean parameter estimates
are in Table~\ref{SimulationParEsts}. We see that the expected MLE
estimates are in many cases quite far from the optimal values (and the
empirical mean for the simulations are also far from optimal). As a
consequence, we expect using MLE to estimate hyperparameter values to
lead to substantially worse ranking than using the optimal values.

\begin{table}[htbp]

\caption{Parameter values. Left: optimal parameter values (repeated
  from Table~\ref{SimExpLoss}(a)) that
  minimise expected loss over top 10\% of data. Right: expected MLE
  estimates for parameter values estimated from truncated
  data.}\label{TheoEstPar}

\begin{subtable}{0.49\textwidth}
\caption{Optimal Parameter Values}
\begin{tabular}{ll|lll}
\hline
&&\multicolumn{3}{c}{Estimating Prior}\\
\cline{3-5}
&&Normal & Exp. & Pareto \\
\hline
\multicolumn{1}{c|}{\multirow{3}{2em}{True Prior}} & Normal & 1 & 1.5614 & 1.2898\\
\multicolumn{1}{c|}{} & Exponential & 1.7005 & 1 &  1.6772 \\
\multicolumn{1}{c|}{} & Pareto & 1.1785 & 1.5811 & 2 \\
\hline
\end{tabular}
\end{subtable}
\begin{subtable}{0.49\textwidth}
\caption{Expected Parameter Estimates}\label{ExpParEsts}
\begin{tabular}{ll|lll}
\hline
&&\multicolumn{3}{c}{Estimating Prior}\\
\cline{3-5}
&&Normal & Exp. & Pareto \\
\hline
\multicolumn{1}{c|}{\multirow{3}{2em}{True Prior}} & Normal & 1 & 2.1122
 & 3.47 \\
\multicolumn{1}{c|}{} & Exponential & 2.5701 & 1 & 3.15  \\
\multicolumn{1}{c|}{} & Pareto & NA  & 0.6325 & 2 \\
\hline
\end{tabular}
\end{subtable}
\end{table}

\subsection{Simulation Results}

The results of the simulation are shown in
Table~\ref{SimResults002}. This table gives the average of loss
function from Equation~(\ref{LossFunctionEqn}) over the simulated
datasets, for each scenario. As expected, with optimal parameter
estimates, using the normal estimating prior when the true prior is
heavy-tailed causes a bigger loss, relative to the difficulty of the
problem (measured as the loss arising from using a point estimate),
than using a heavy-tailed estimating prior when the true prior is
normal --- when the true prior is normal, the problem is much more
difficult (the increase in loss from using the point estimate is
larger), but the increase in loss from using the Pareto estimating
prior is about the same as the increase when using a normal estimating
prior in the easier case where the true prior follows a Pareto
distribution. When we use estimated hyperparameter values, the loss
from using the Pareto estimating prior when the true prior is normal
is larger than using a normal estimating prior when the true prior is
Pareto, even taken relative to the loss from using a point
estimate. This is explained by the fact that the MLE estimate for the
Pareto parameter is further from the optimal value than the MLE
estimate of $\frac{1}{\tau^2}$ is from it's optimal value. Using an
exponential estimating prior does not perform too badly in any of the
cases. All methods perform much better than the use of the point
estimates. These results show a similar result to the theoretically
estimated values in Table~\ref{SimExpLoss}(b) with many values
approximately proportional to that table. The error in the case when
the estimating prior is normal and the true prior is heavy-tailed, is
theoretically bounded because the Pareto distribution has
$\alpha>0.5$, but results are still poor.

\begin{table}[htbp]

\caption{Simulation Results: average over simulated data sets of loss
  function (from Equation~\ref{LossFunctionEqn}). Left tables use
  optimal parameter values. Right tables use MLE estimated parameter
  values (for data truncated at the true 90th percentile). Top row is
  for sample size 1000 (1000 datasets), middle row sample size 10000
  (100 datasets), bottom row 100000 (10 datasets). Mean of $\sigma$ is
  0.02. }\label{SimResults002}

\begin{subtable}{0.48\textwidth}
\caption{Theoretical, sample size 1000}
\begin{tabular}{ll|lllr}
\hline
&&\multicolumn{3}{c}{Estimating Prior}&Point \\
\cline{3-5}
&&Normal & Exp. & Pareto & Estimate \\
\hline
\multicolumn{1}{c|}{\multirow{3}{2em}{True Prior}} & Normal & 0     &
0.003 & 0.008 & 0.038 \\
\multicolumn{1}{c|}{} & Exp. & 0 &     0 &     0 &  0.009 \\
\multicolumn{1}{c|}{} & Pareto & 0.003 &     0 &     0 & 0.022 \\
\hline
\end{tabular}
\end{subtable}
\begin{subtable}{0.48\textwidth}
\caption{Estimated, sample size 1000}
\begin{tabular}{ll|rrrr}
\hline
&&\multicolumn{3}{c}{Estimating Prior}&Point \\
\cline{3-5}
&&Normal & Exp. & Pareto & Estimate \\
\hline
\multicolumn{1}{c|}{\multirow{3}{2em}{True Prior}} & Normal & 0.000     & 0.008 & 0.039 & 0.038 \\
\multicolumn{1}{c|}{} & Exp. & 0.001 & 0.000 & 0.003 & 0.009 \\
\multicolumn{1}{c|}{} & Pareto & 0.017 & 0.008 & 0.000 & 0.022 \\
\hline
\end{tabular}
\end{subtable}
\begin{subtable}{0.48\textwidth}
\caption{Theoretical, sample size 10000}
\begin{tabular}{ll|lllr}
\hline
&&\multicolumn{3}{c}{Estimating Prior}&Point \\
\cline{3-5}
&&Normal & Exp. & Pareto & Estimate \\
\hline
\multicolumn{1}{c|}{\multirow{3}{2em}{True Prior}} & Normal &    0 & 0.12 & 0.35 & 3.82 \\
\multicolumn{1}{c|}{} & Exp. & 0.01 &    0 & 0.02 & 0.75 \\
\multicolumn{1}{c|}{} & Pareto & 0.31 & 0.05 &    0 & 2.03 \\
\hline
\end{tabular}
\end{subtable}
\begin{subtable}{0.48\textwidth}
\caption{Estimated, sample size 10000}
\begin{tabular}{ll|rrrr}
\hline
&&\multicolumn{3}{c}{Estimating Prior}& Point \\
\cline{3-5}
&&Normal & Exp. & Pareto & Estimate \\
\hline
\multicolumn{1}{c|}{\multirow{3}{2em}{True Prior}} & Normal & 0.00 &
\phantom{1}0.61 & 3.57 & 3.82 \\
\multicolumn{1}{c|}{} & Exp. & 0.07 &    0.00 & 0.25 & 0.75 \\
\multicolumn{1}{c|}{} & Pareto & 1.74 & 0.73 & 0.00 & 2.03
\\
\hline
\end{tabular}
\end{subtable}
\begin{subtable}{0.52\textwidth}
\caption{Theoretical, sample size 100000}
\begin{tabular}{ll|rrrr}
\hline
&&\multicolumn{3}{c}{Estimating Prior}&Point \\
\cline{3-5}
&&Normal & Exp. & Pareto & Estimate\\
\hline
\multicolumn{1}{c|}{\multirow{3}{2em}{True Prior}} & Normal &  0\phantom{.0} &
10.1 & \phantom{1}32.7 & 385.1 \\
\multicolumn{1}{c|}{} & Exp. &  2.0 &    0\phantom{.0} &    2.0 & 76.5 \\
\multicolumn{1}{c|}{} & Pareto & 33.0 &  5.9 &    0\phantom{.0} & 209.0 \\
\hline
\end{tabular}
\end{subtable}
\begin{subtable}{0.48\textwidth}
\caption{Estimated, sample size 100000}
\begin{tabular}{ll|rrrr}
\hline
&&\multicolumn{3}{c}{Estimating Prior}&Point \\
\cline{3-5}
&&Normal & Exp. & Pareto & Estimate \\
\hline
\multicolumn{1}{c|}{\multirow{3}{2em}{True Prior}} & Normal &   0.0 &
\phantom{1}59.5 & 353.6 & 385.1 \\
\multicolumn{1}{c|}{} & Exp. & 8.2 &  0.1 &  27.1 & 76.5 \\
\multicolumn{1}{c|}{} & Pareto & 187.0 & 77.8 &  0.1 & 209.0
 \\
\hline
\end{tabular}
\end{subtable}

\end{table}

Table~\ref{SimulationParEsts} gives the mean parameter estimates in
cases where we used MLE to estimate parameter values. We see that
these are mostly as predicted in Table~\ref{ExpParEsts}(b). The main
difference is when we use a normal estimating prior generated under an
exponential true prior. Here the estimated value is much closer to the
optimal value. This is because the approximation we used in deriving
the expected MLE estimate is not very accurate. This explains why the
normal estimating prior with estimated parameter values did not
perform so poorly in this scenario. We know that the ranking based on
a normal estimating prior is most sensitive to parameter
estimates. However, because MLE provides a fairly good estimate in
this case, the loss from using an estimated value is not so great. For
the exponential and Pareto estimating priors, the MLE does not provide
good parameter estimates for the purpose of ranking. Because the
ranking loss in these cases is less sensitive to estimation errors in
the hyperparameters, the resulting losses are not excessive. However,
this indicates there is great scope for improving results by devising
better parameter estimation techniques. It is also worth noting that
these hyperparameters were estimated to fit the tail well, rather than
the whole dataset. More common practice is to estimate the
hyperparameters based on the whole data. We would expect this to
result in much worse ranking results, particularly for the normal
estimating prior where the loss is particularly sensitive to the parameter
estimates.

\begin{table}[htbp]
\caption{parameter estimates}\label{SimulationParEsts}

\hfil\begin{tabular}{lllllll}
\hline
True Prior & Sample size & Normal & Exponential &
Pareto \\
\hline
\hline
\multirow{3}{4em}{Normal}
 & 1000   & 0.999(0.0883) & 2.135(0.1902) & 3.479(0.2631) \\
 & 10000  & 1.000(0.0186) & 2.120(0.0586) & 3.462(0.0811) \\
 & 100000 & 1.000(0.0061) & 2.119(0.0192) & 3.460(0.0265) \\
\cline{2-5}
\multirow{3}{4em}{Exponential}
 & 1000   & 2.051(0.1539) & 1.012(0.1031) & 3.113(0.2538) \\
 & 10000  & 2.055(0.0485) & 1.002(0.0316) & 3.094(0.0780) \\
 & 100000 & 2.055(0.0150) & 1.001(0.0097) & 3.093(0.0236) \\
\cline{2-5}
\multirow{3}{4em}{Pareto}
 & 1000   & 36.081(232.50) & 0.674(0.1494) & 2.024(0.2059) \\
 & 10000  & 80.509(445.18) & 0.639(0.0589) & 2.005(0.0640) \\
 & 100000 & 99.292(487.59) & 0.634(0.0225) & 2.003(0.0200) \\
\hline
\end{tabular}
\end{table}

\section{Real Data Analysis}\label{SectionRealData}

\subsection{Type 2 Diabetes}

We look at several real data sets. These datasets were studied by
Henderson \& Newton (2015) for their work on $r$-values.  The first
data set consists of GWAS data for log odds ratio between SNPs and
type-2 diabetes from Morris {\it et al.}~(2006). The data are
available from \url{http://diagram-consortium.org/downloads.html}. The
data consists of 137,899 SNPs from 12,171 type 2 diabetes cases and
56,862 controls. For each SNP, an odds ratio is available along with a
95\% confidence interval. Following Henderson \& Newton (2015), we
have taken the value as the log-odds ratio, assuming this estimate
follows a normal distribution, and that the standard deviation of this
distribution is one-quarter of the width of the log of the 95\%
confidence interval provided. The resulting positive data points and
isotaxes for a normal and exponential estimating prior are shown in
Figure~\ref{DiabetesIsotaxes}.

\begin{figure}[htbp]

\begin{subfigure}{0.35\textwidth}
\caption{Exponential Estimating Prior}
\includegraphics[width=5.6cm]{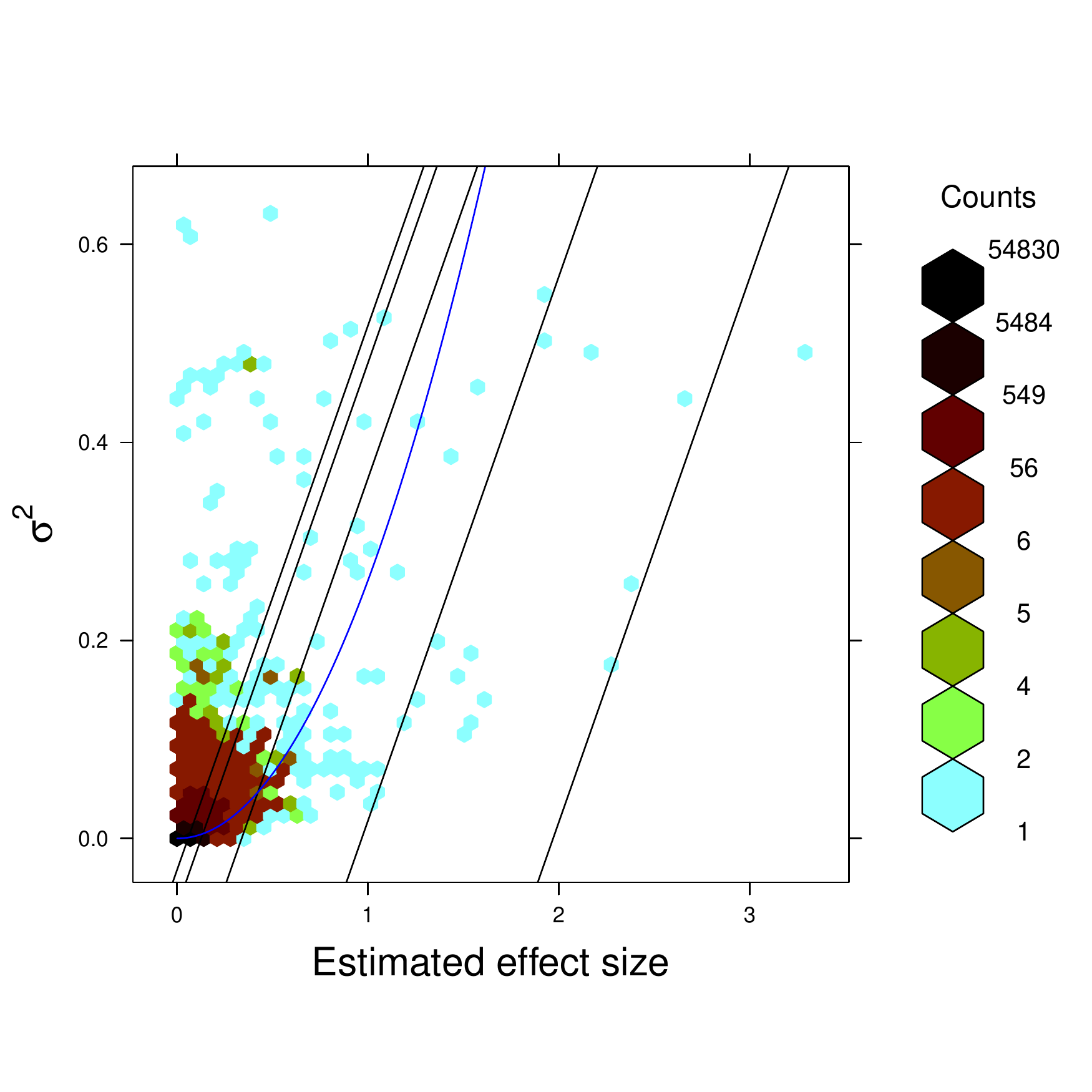} 
\end{subfigure}
\begin{subfigure}{0.35\textwidth}
\caption{Normal Estimating Prior, $\tau^2=0.001811156$ (estimated from full data
  by method of moments)}
\includegraphics[width=5.6cm]{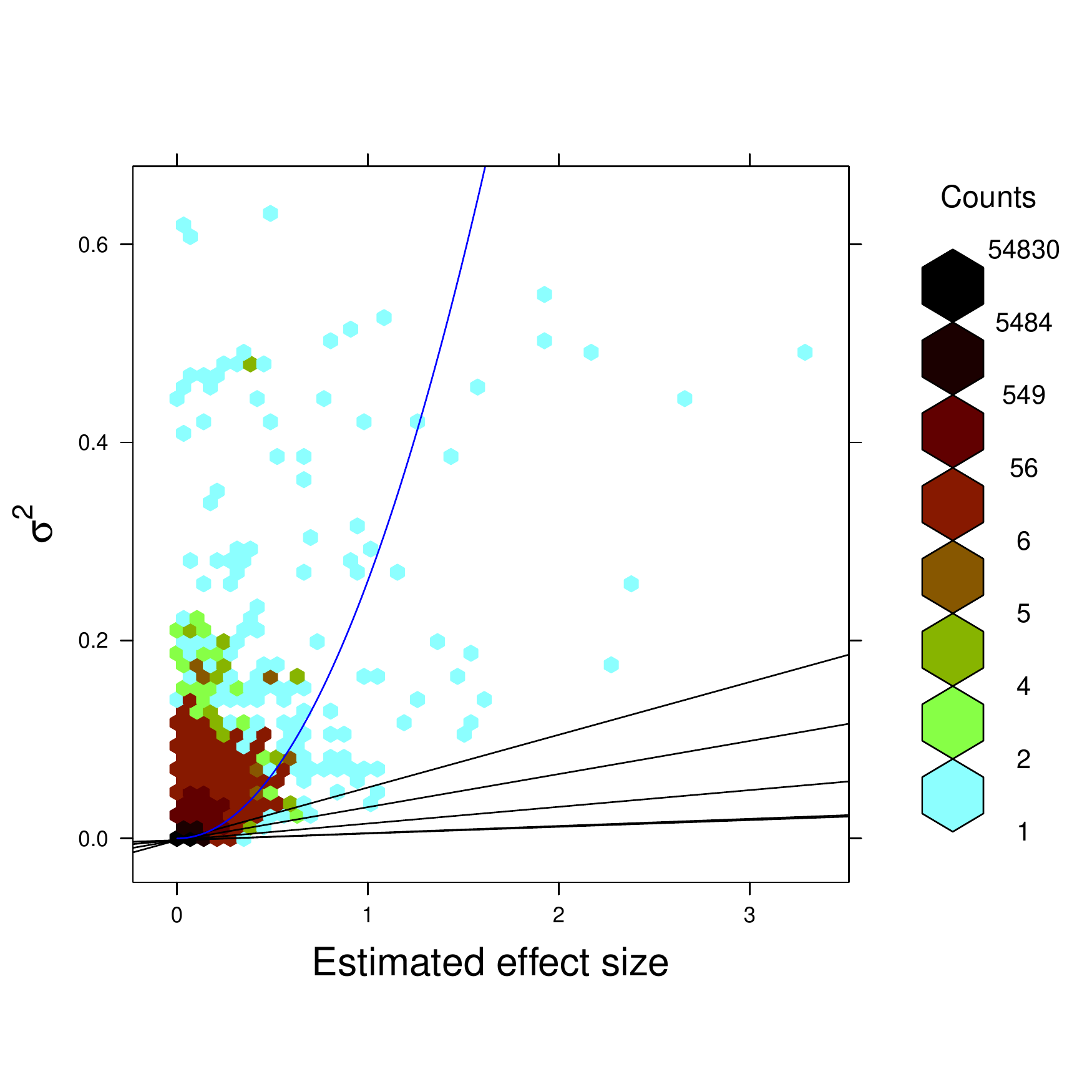}
\end{subfigure}
\begin{subfigure}{0.35\textwidth}
\caption{Normal Estimating Prior, $\tau^2=0.08$}\label{DiabetesIsotaxes_c}
\includegraphics[width=5.6cm]{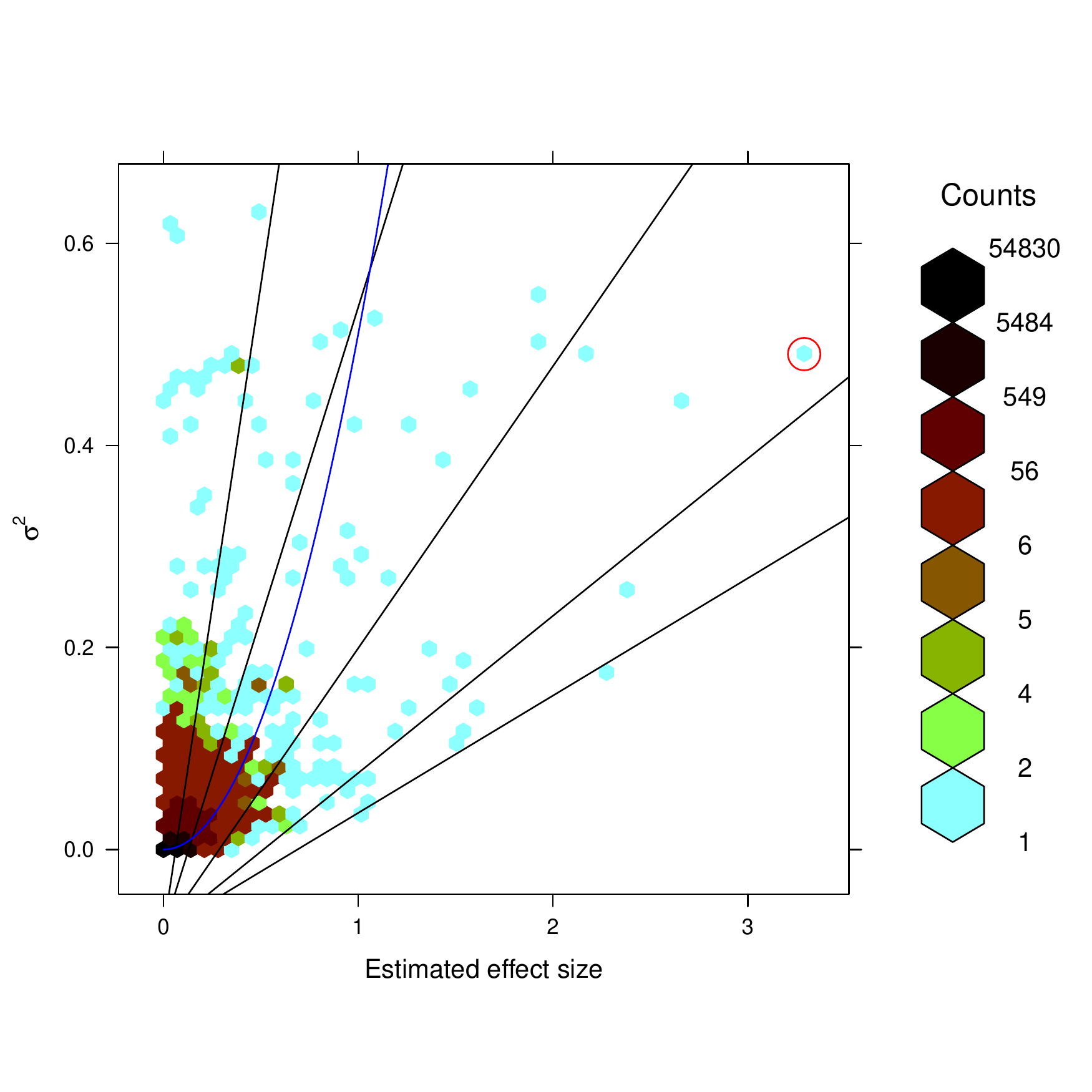}
\end{subfigure}

\caption{Isotaxes for exponential and normal estimating prior
  distributions for type-2 diabetes data. Isotaxes shown are for top
  5\%, top 1\%, top 0.1\%, top 0.01\% and top 0.001\%. The blue curve
  represents the 95\% significance level against the null hypothesis
  --- that is, points to the left of the blue curve do not have
  estimated effect size significantly different from
  zero.}\label{DiabetesIsotaxes}

\end{figure}

From this figure, we see that using a normal estimating prior with
naively estimated variance, the estimated variance is small, leading
to isotaxes passing close to the origin. This makes the ranking focus
on values with small variance, and rank values with larger observed
value and larger variance behind values with smaller variance. The
exponential estimating prior provides a ranking that selects many more
of the points with large estimated effect size. We can improve the
performance of a normal prior by artificially inflating the variance
to match the tail better. Figure~\ref{DiabetesIsotaxes_c} shows the
effect of this. This does select a lot more of the points with large
effect sizes, and indeed the top 1\% isotaxis is very similar to the
1\% isotaxis for the exponential estimating prior. On the other hand,
the higher isotaxes do put too much weight on having smaller standard
error, ranking a number of points with smaller estimated effect size
ahead of the point with largest estimated effect size (circled in
red). It is extremely implausible that this ranking is
correct. Overall, the ranking based on an exponential estimating prior
appears more plausible to us for this dataset.

\subsection{Breast Cancer}

Next we look at the gene expression data relating to breast cancer
from West {\it et al.} (2001). This dataset is available in the {\tt
  rvalues} package. The data set consists of gene expression
measurements of 7129 genes accross 49 breast tumour samples --- 25
oestrogen receptor (ER)+ samples and 24 ER$-$ samples. For each gene,
the difference in means between the ER$+$ and ER$-$ groups is
calculated, along with its appropriate standard error. In theory the
error distribution should be modelled as a $t$-distribution with 47
degrees of freedom. However, for the purpose of this paper, we have
used a normal distribution. The loss of accuracy should be fairly
small. The resulting plot of variance against estimated effect size,
and isotaxes for posterior mean with an exponential and a normal
estimating prior are shown in Figure~\ref{BCIsotaxes}.

\begin{figure}[htbp]

\begin{subfigure}{0.35\textwidth}
\caption{Exponential Estimating Prior\\\vphantom{Ey}}
\includegraphics[width=5.6cm]{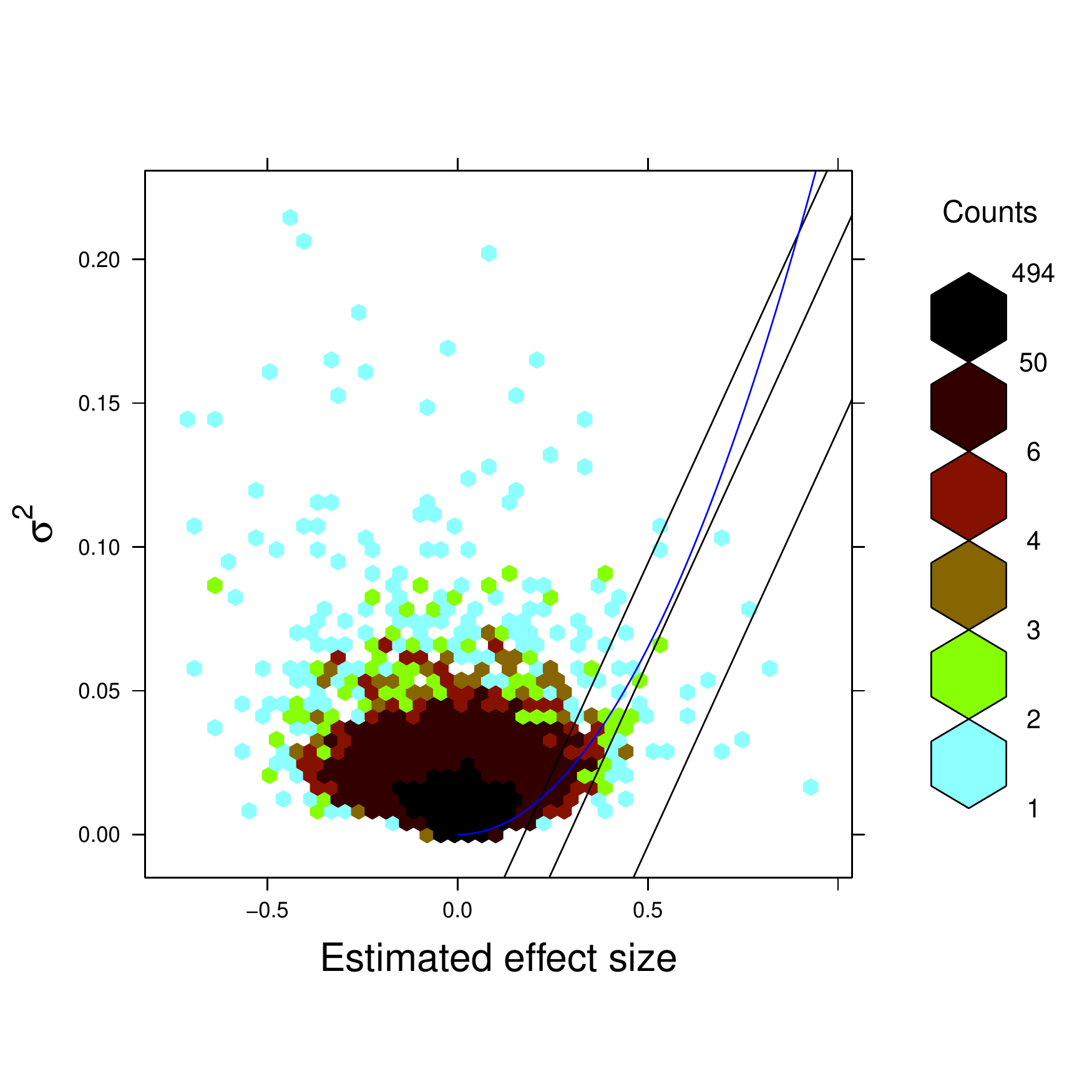}
\end{subfigure}
\begin{subfigure}{0.35\textwidth}
\caption{Normal Estimating Prior, $\tau^2=0.002478642$, estimated from data by
  method of moments}
\includegraphics[width=5.6cm]{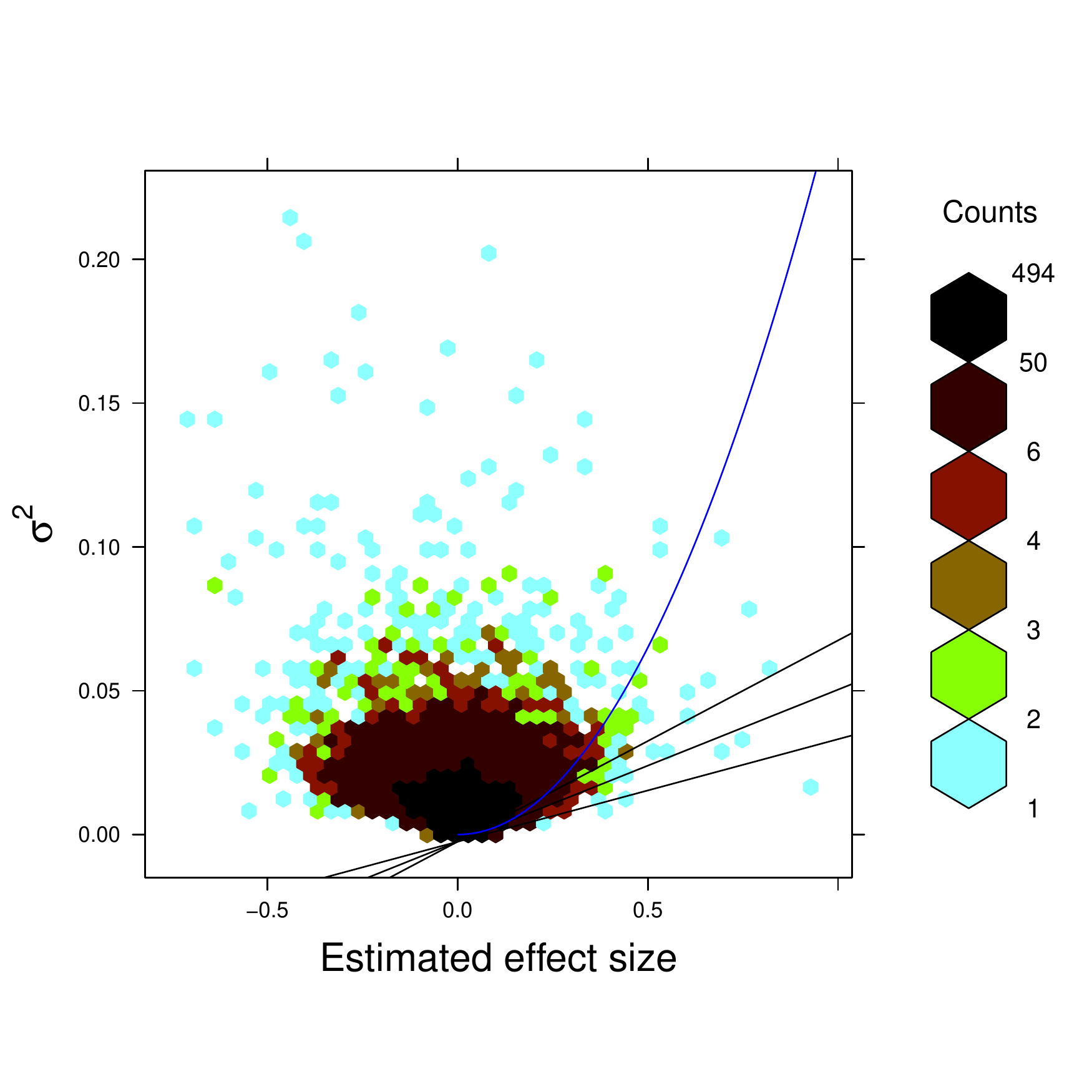}
\end{subfigure}
\begin{subfigure}{0.35\textwidth}
\caption{Normal Estimating Prior, $\tau^2=0.08$\\\vphantom{Ey}}
\includegraphics[width=5.6cm]{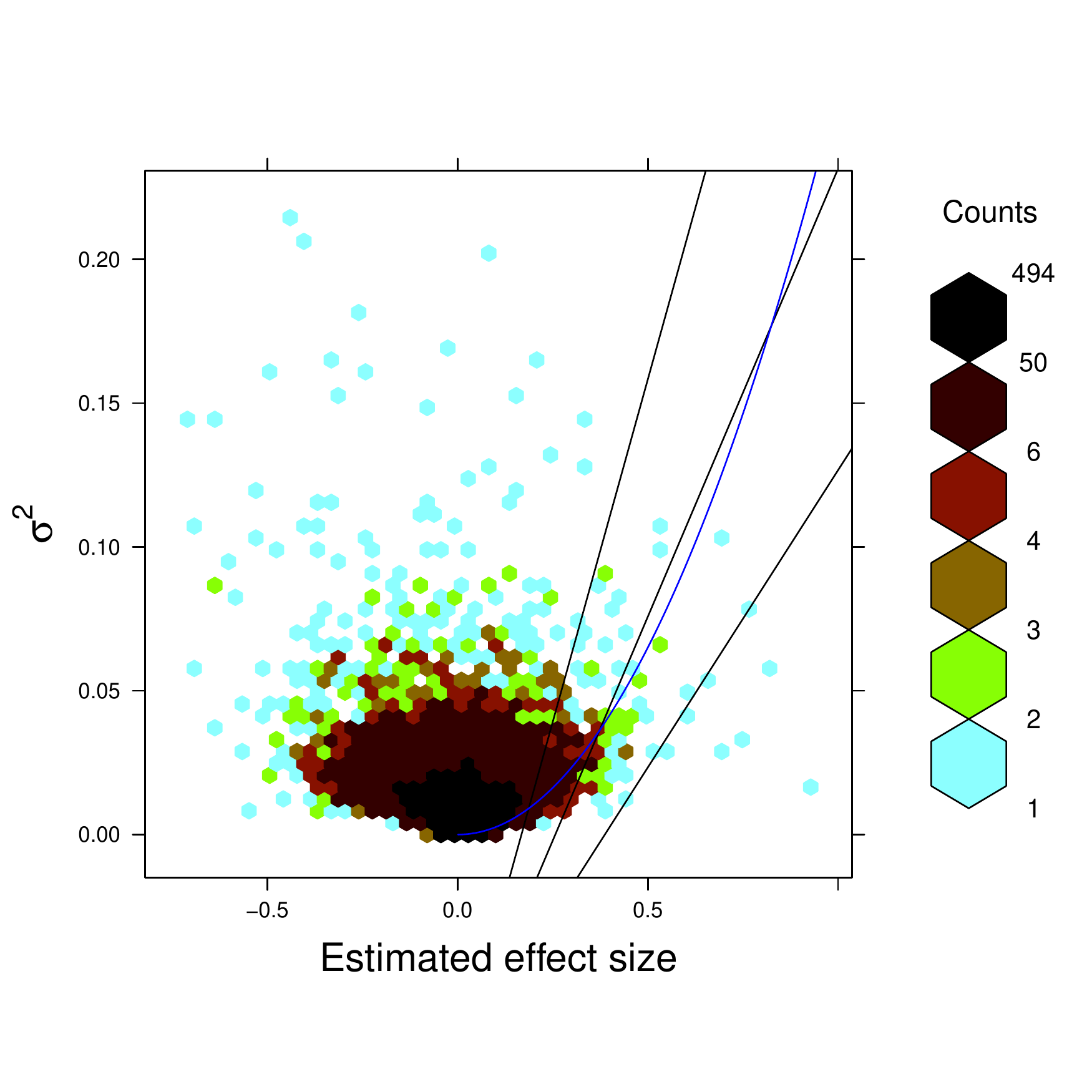}
\end{subfigure}

\caption{Isotaxes for exponential and normal estimating prior
  distributions for breast cancer data. Isotaxes shown are for top
  5\%, top 1\%, and top 0.1\%. The blue curve represents the 95\%
  significance level against the null hypothesis --- that is, points
  to the left of the blue curve do not have estimated effect size
  significantly more than zero using a two-sided
  test.}\label{BCIsotaxes}

\end{figure}

As for the diabetes data, we see that the normal estimating prior
results in very flat isotaxes, and therefore gives a high ranking to
observations with small variance. Meanwhile, the exponential
estimating prior puts a lot more weight on points with large estimated
effect size. Again, we see that using a normal estimating prior with
inflated variance results in isotaxes that are more similar to the
exponential estimating prior. In this case, the differences between
the two rankings are not so clear-cut as the previous case, where some
of the rankings using a normal estimating prior with increased
variance were completely implausible. In this case, both the rankings
for the exponential estimating prior and the normal estimating prior
with inflated variance seem reasonable.

\section{Conclusions and Future Work}\label{SectionConclusions}

We have seen that choice of estimating prior can have a very large
effect on Bayesian ranking methods. For the majority of ranking
problems, we are particularly interested in ranking at the upper tail
of the distribution. The ranking of the upper tail can be particularly
affected by choice of estimating prior.

Using a light-tailed estimating prior for posterior mean ranking can
lead to very bad results. If the true prior is heavy-tailed, the
posterior mean can be far away from the truth. Conversely, if the
estimating prior is too heavy-tailed, the posterior mean estimated
will be between the posterior mean under the true prior and the point
estimate. This cannot be too far from the true posterior mean. This
means that using an exponential or heavier-tailed distribution as the
estimating prior should be more robust to model misspecification.

In addition to being less robust to model misspecification,
light-tailed estimating priors can be more sensitive to estimated
parameter values. In cases where the estimating prior is misspecified,
using MLE estimates for hyperparameters can also be far from optimal,
so this can lead to bad results even in cases with large
datasets. Since we are usually particularly interested in the top
units, it is usually advisable to choose parameter values that fit the
tail of the distribution well.

Using a non-parametric prior is robust, in that there is an upper
bound on how far the posterior mean can be from the posterior mean
under the true prior. However, using a non-parametric prior can be
inefficient for smaller sample sizes, and can lead to some strange
rankings.

We have confirmed our results by simulation studies and real data
examples. In the simulation study, we found that an exponential
estimating prior performed relatively well regardless of the true
prior. Our simulation study also studied the effect of estimating
hyperparameters on the performance. As expected, estimating
hyperparameters does cause some loss. The estimation in this
simulation was done by maximum likelihood. However, since the loss
function we are aiming to minimise is not the standard squared error
loss, this is not the optimal estimation method. We based our
parameter estimation on the upper 10\% of the data points. It is
common for analyses which use the whole data to estimate
hyperparameters. Doing this could lead to far worse results when the
estimating prior is misspecified.

Overall, unless there is good evidence otherwise, we suggest an
exponential estimating prior will be a good compromise between robustness and
efficiency in most cases. It also offers easy computation of posterior
mean.

\subsection{Future Work}

The most obvious direction to need improvement in this research is
hyperparameter estimation. We have seen in our simulation study that
estimation by MLE can lead to bad ranking results. This is because the
loss function from prior misspecification is different from the loss
function that MLE estimation aims to minimise. This suggests that a
different method of estimating the hyperparameters is needed --- a
method specifically targeted at optimising ranking estimates. We know
the loss function that we are aiming to minimise, so it should be
possible to find an explicit way to solve this and derive a procedure
for estimating the hyperparameter. Given our recommendation to use an
exponential estimating prior in most cases, finding the best
hyperparameters should not prove too challenging a problem.

Our study has a number of limitations. We have considered only cases
where ranking is by posterior mean and the error distribution is
normal. In future work, we should study the problem for different
error distributions, not just normal. Further estimation is also
needed into cases where the variance of the error distribution depends
on the parameter $\theta$. This can allow certain approximations to be
applied. For example a Poisson distribution can be approximated by a
normal distribution where the variance depends on the mean. We should
also study the problem for different methods and objective functions,
e.g. $r$-values, posterior expected rank.

We also have not considered the effect of model selection on
ranking. If the estimating prior distribution is chosen based on
certain model selection criteria, this may improve the
ranking. However, model selection for mixture models can be difficult,
so it might not provide the improvements we hope for. Model selection
also depends upon a good set of candidate models. Our research
suggests that the form of the function
$\lambda(x)=-\frac{\pi'(x)}{\pi(x)}$ is most crucial in our choice of
estimating prior, so including a sufficient range of models to allow
flexibility in this function should allow us to obtain good ranking
results, provided the model selection criteria are well chosen to be
related to our objective function.


\appendix

\section{Loss Function Calculation}
\label{AppProofOfLossFunction}

\begin{proposition}
Suppose the true prior distribution of the parameter $\theta$ has
density function $\pi(\theta)$, and that we have two observations
$x_1$ and $x_2$ which are normally distributed with means $\theta_1$
and $\theta_2$ and standard deviations $\sigma_1$ and $\sigma_2$
respectively, where $\theta_1$ and $\theta_2$ are random samples from
the true prior distribution, and $\sigma_1$ and $\sigma_2$ are assumed
to be small.

\begin{enumerate}
\renewcommand\labelenumi{(\roman{enumi})}
\renewcommand\theenumi\labelenumi

\item The expected loss when the estimating prior and the true prior
  are the same (which we will refer to as the {\em optimal expected
    loss}) is approximately given by
$$\frac{\sigma_1{}^2+\sigma_2{}^2}{2}{\mathbb
  E}(\pi(x))$$

\item When the estimating prior has density $\hat{\pi}$, the
  difference between the expected loss and the optimal expected loss
  is approximately given by
\begin{equation*}\frac{1}{2}(\sigma_1{}^2-\sigma_2{}^2)^2\int_{a}^\infty
  \pi(x)^2(\lambda(x)-\hat{\lambda}(x))^2\,dx\end{equation*} 
where $\hat{\lambda}(x)=-\frac{\hat{\pi}'(x)}{\hat{\pi}(x)}$.

\item The difference between the expected loss from using the point
  estimate $x_i$ and the optimal expected loss is approximately given by

$$\frac{1}{2}\left(\sigma_1{}^2-\sigma_2{}^2\right)^2\int_a^\infty \pi'(x)^2\,dx$$

\end{enumerate}
\end{proposition}

\begin{proof}

(i) Suppose that the true parameter values are $\theta_1$ and
  $\theta_2$ respectively. Let $\Delta=\theta_1-\theta_2$. Now the
  loss from mis-ranking is $|\Delta|$ if the points are mis-ranked and
  0 if they are not misranked. The points are misranked if either
  $\Delta>0$ and
  $x_1-\lambda(x_1)\sigma_1{}^2<x_2-\lambda(x_2)\sigma_2{}^2$ or if
  $\Delta<0$ and
  $x_1-\lambda(x_1)\sigma_1{}^2>x_2-\lambda(x_2)\sigma_2{}^2$. Since
  the points will not plausibly be misranked if $\Delta$ is large
  (since $x_1$ and $x_2$ will then with high probability be far
  apart), we will assume that $\Delta$ is small, so that we have
  $\lambda(x_1)\approx\lambda(x_2)\approx\lambda(\theta_1)$. We will
  denote this common value $\lambda$. Now for fixed $\theta_1$ and
  $\theta_2$, suppose $\Delta>0$; we want to calculate
  $P\left(x_1-\lambda\sigma_1{}^2<x_2-\lambda\sigma_2{}^2\right)$. We
  know that $x_1-x_2$ is normally distributed with mean $\Delta$ and
  variance $\sigma_1{}^2+\sigma_2{}^2$. Therefore,
$$P\left(x_1-\lambda\sigma_1{}^2<x_2-\lambda\sigma_2{}^2\right)=\Phi\left(\frac{\lambda(\sigma_1{}^2-\sigma_2{}^2)-\Delta}{\sqrt{\sigma_1{}^2+\sigma_2{}^2}}\right)$$
On the other hand, if $\Delta<0$, we have  
$$P\left(x_1-\lambda\sigma_1{}^2>x_2-\lambda\sigma_2{}^2\right)=\Phi\left(\frac{-\lambda(\sigma_1{}^2-\sigma_2{}^2)+\Delta}{\sqrt{\sigma_1{}^2+\sigma_2{}^2}}\right)$$
Now suppose that we fix $\theta_1$, and we want to take the expected
loss over the distribution of $\theta_2$. This is given by 
\begin{align*}
l(\theta_1)&=\int_0^\infty
\pi(\theta_1-\Delta)\Delta\Phi\left(\frac{\lambda(\sigma_1{}^2-\sigma_2{}^2)-\Delta}{\sqrt{\sigma_1{}^2+\sigma_2{}^2}}\right)\,d\Delta-\int_{-\infty}^0
\pi(\theta_1-\Delta)\Delta\Phi\left(\frac{-\lambda(\sigma_1{}^2-\sigma_2{}^2)+\Delta}{\sqrt{\sigma_1{}^2+\sigma_2{}^2}}\right)\,d\Delta\\
\end{align*}
Since the probability of misranking is negligible for large $\Delta$,
we will consider only small values of $\Delta$. For these small
values, we can take the Taylor expansion 
$$\pi(\theta_1-\Delta)=\pi(\theta_1)-\Delta\pi'(\theta_1)=\pi(\theta_1)\left(1+\Delta\lambda(\theta_1)\right)$$
Substituting this into the above loss function gives
\begin{align*}
l(\theta_1)&=\int_0^\infty
\pi(\theta_1)\left(1+\Delta\lambda(\theta_1)\right)\Delta\Phi\left(\frac{\lambda(\sigma_1{}^2-\sigma_2{}^2)-\Delta}{\sqrt{\sigma_1{}^2+\sigma_2{}^2}}\right)\,d\Delta-\int_{-\infty}^0
\pi(\theta_1)\left(1+\Delta\lambda(\theta_1)\right)\Delta\Phi\left(\frac{-\lambda(\sigma_1{}^2-\sigma_2{}^2)+\Delta}{\sqrt{\sigma_1{}^2+\sigma_2{}^2}}\right)\,d\Delta\\
\end{align*}

We recall that 

\begin{align*}
\int_{-c}^\infty \frac{\xi^3
e^{-\frac{\xi^2}{2\sigma^2}}}{\sqrt{2\pi}\sigma}\,d\xi&=\left[-\frac{\sigma}{\sqrt{2\pi}}\xi^2e^{-\frac{\xi^2}{2\sigma^2}}\right]_{-c}^\infty+\int_{-c}^\infty
  \frac{2\sigma}{\sqrt{2\pi}}\xi e^{-\frac{\xi^2}{2\sigma^2}}\,d\xi\\
&=\frac{1}{\sqrt{2\pi}}(\sigma c^2+2\sigma^3)e^{-\frac{c^2}{2\sigma^2}}\\
\int_{-c}^\infty \frac{\xi^2
e^{-\frac{\xi^2}{2\sigma^2}}}{\sqrt{2\pi}\sigma}\,d\xi&=\left[-\frac{\sigma}{\sqrt{2\pi}}\xi
  e^{-\frac{\xi^2}{2\sigma^2}}\right]_{-c}^\infty+\int_{-c}^\infty
  \frac{\sigma}{\sqrt{2\pi}} e^{-\frac{\xi^2}{2\sigma^2}}\,d\xi\\
&=\sigma^2\Phi\left(\frac{c}{\sigma}\right)-\frac{\sigma c
    e^{-\frac{c^2}{2\sigma^2}}}{\sqrt{2\pi}}\\
\int_{-c}^\infty \frac{\xi
e^{-\frac{\xi^2}{2\sigma^2}}}{\sqrt{2\pi}\sigma}\,d\xi&=\frac{\sigma e^{-\frac{c^2}{2\sigma^2}}}{\sqrt{2\pi}}\\
\int_{-c}^\infty
\frac{e^{-\frac{\xi^2}{2\sigma^2}}}{\sqrt{2\pi}\sigma}\,d\xi
&=\Phi\left(\frac{c}{\sigma}\right)\\
\end{align*}

Hence we calculate

\begin{align*}
\int_{0}^\infty
\Delta^2\Phi\left(\frac{c-\Delta}{\sigma}\right)\,d\Delta&=\left[\frac{\Delta^3}{3}\Phi\left(\frac{c-\Delta}{\sigma}\right)\right]_0^\infty+\int_0^\infty
  \frac{\Delta^3}{3}\frac{e^{-\frac{(c-\Delta)^2}{2\sigma^2}}}{\sqrt{2\pi}\sigma}\,d\Delta\\
&=\int_{-c}^\infty
  \frac{(\xi+c)^3}{3}\frac{e^{-\frac{\xi^2}{2\sigma^2}}}{\sqrt{2\pi}\sigma}\,d\xi\\
&=\frac{1}{3}\left(\frac{1}{\sqrt{2\pi}}(\sigma c^2+2\sigma^3)e^{-\frac{c^2}{2\sigma^2}}+3c\left(\sigma^2\Phi\left(\frac{c}{\sigma}\right)-\frac{\sigma c
    e^{-\frac{c^2}{2\sigma^2}}}{\sqrt{2\pi}}\right)+3c^2\frac{\sigma
    e^{-\frac{c^2}{2\sigma^2}}}{\sqrt{2\pi}}+c^3\Phi\left(\frac{c}{\sigma}\right)\right)\\
&=\frac{1}{3}\left((\sigma c^2+2\sigma^3-3c^2\sigma+3c^2\sigma)\frac{e^{-\frac{c^2}{2\sigma^2}}}{\sqrt{2\pi}}+(3c\sigma^2+c^3)\Phi\left(\frac{c}{\sigma}\right)\right)\\
&=\frac{1}{3}\left((2\sigma^3+\sigma
  c^2)\frac{e^{-\frac{c^2}{2\sigma^2}}}{\sqrt{2\pi}}+(3c\sigma^2+c^3)\Phi\left(\frac{c}{\sigma}\right)\right)\\
\int_{0}^\infty
\Delta\Phi\left(\frac{c-\Delta}{\sigma}\right)\,d\Delta&=\left[\frac{\Delta^2}{2}\Phi\left(\frac{c-\Delta}{\sigma}\right)\right]_0^\infty+\int_0^\infty
  \frac{\Delta^2}{2}\frac{e^{-\frac{(c-\Delta)^2}{2\sigma^2}}}{\sqrt{2\pi}\sigma}\,d\Delta\\
&=\int_{-c}^\infty
  \frac{(\xi+c)^2}{2}\frac{e^{-\frac{\xi^2}{2\sigma^2}}}{\sqrt{2\pi}\sigma}\,d\xi\\
&=\frac{1}{2}\left(\left(\sigma^2\Phi\left(\frac{c}{\sigma}\right)-\frac{\sigma c
    e^{-\frac{c^2}{2\sigma^2}}}{\sqrt{2\pi}}\right)+2c\frac{\sigma
    e^{-\frac{c^2}{2\sigma^2}}}{\sqrt{2\pi}}+c^2\Phi\left(\frac{c}{\sigma}\right)\right)\\
&=\frac{1}{2}\left((\sigma^2+c^2)\Phi\left(\frac{c}{\sigma}\right)+\frac{\sigma c
    e^{-\frac{c^2}{2\sigma^2}}}{\sqrt{2\pi}}\right)\\
\end{align*}

In the loss function, we let $c=\lambda(\sigma_1{}^2-\sigma_2{}^2)$
and $\sigma=\sqrt{\sigma_1{}^2+\sigma_2{}^2}$. Substituting these into
the loss function gives:

\begin{align*}
l(\theta_1)&=\pi(\theta_1)\left(\int_0^\infty
\left(\Delta+\Delta^2\lambda\right)\Phi\left(\frac{c-\Delta}{\sigma}\right)\,d\Delta-\int_{-\infty}^0
\left(\Delta+\Delta^2\lambda\right)\Phi\left(\frac{-c+\Delta}{\sigma}\right)\,d\Delta\right)\\
&=\pi(\theta_1)\left(\int_0^\infty
\left(\Delta+\Delta^2\lambda\right)\Phi\left(\frac{c-\Delta}{\sigma}\right)\,d\Delta+\int_0^{\infty}
\left(\Delta-\Delta^2\lambda\right)\Phi\left(\frac{-c-\Delta}{\sigma}\right)\,d\Delta\right)\\
&=\pi(\theta_1)\left(\int_0^\infty
\Delta\Phi\left(\frac{c-\Delta}{\sigma}\right)\,d\Delta+\int_0^{\infty}
\Delta\Phi\left(\frac{-c-\Delta}{\sigma}\right)\,d\Delta+\lambda\int_0^\infty\Delta^2\Phi\left(\frac{c-\Delta}{\sigma}\right)\,d\Delta-\lambda\int_0^\infty\Delta^2\Phi\left(\frac{-c-\Delta}{\sigma}\right)\,d\Delta\right)\\
&=\pi(\theta_1)\left(\frac{1}{2}\left((\sigma^2+c^2)\left(\Phi\left(\frac{c}{\sigma}\right)+\Phi\left(\frac{-c}{\sigma}\right)\right)+\frac{\sigma e^{-\frac{c^2}{2\sigma^2}}}{\sqrt{2\pi}}(c+(-c))\right)+\frac{\lambda}{3}\left((3c\sigma^2+c^3)\Phi\left(\frac{c}{\sigma}\right)+(3c\sigma^2+c^3)\Phi\left(\frac{-c}{\sigma}\right)\right)\right)\\
&=\pi(\theta_1)\left(\frac{1}{2}\left((\sigma^2+c^2)\right)+\frac{\lambda}{3}\left((3c\sigma^2+c^3)\left(\Phi\left(\frac{c}{\sigma}\right)+\Phi\left(-\frac{c}{\sigma}\right)\right)\right)\right)\\
&=\pi(\theta_1)\left(\frac{1}{2}\left((\sigma^2+c^2)\right)+\frac{\lambda(3c\sigma^2+c^3)}{3}\right)\\
\end{align*}

If we let $d=\sigma_1{}^2-\sigma_2{}^2$, so that $c=\lambda d$, then
we have

\begin{align*}
l(\theta_1)&=\pi(\theta_1)\left(\frac{1}{2}\left((\sigma^2+d^2\lambda^2)\right)+\frac{\lambda(3d\lambda\sigma^2+d^3\lambda^3)}{3}\right)\\
\end{align*}

(For anyone thinking at this point that the dimensions do not work in
this formula, it is worthwhile to remember that $\lambda$ and $\pi$
are inversely proportional to changes in the scale of $\theta$. That
is, if we change the units so that the value of $\theta$ doubles, the
values of $\lambda$ and $\pi$ will be halved. )

With $\frac{d\lambda}{\sigma}$ assumed to be small, we can neglect the
$d^3\lambda^4$ term to get 
\begin{align*}
l(\theta_1)&=\pi(\theta_1)\left(\frac{1}{2}\left(\sigma^2+d^2\lambda^2\right)+d\lambda^2\sigma^2\right)\\
\end{align*}
If we assume that $d^2$ is
negligible, then our expression becomes 
\begin{align*}
l(\theta_1)&=\pi(\theta_1)\sigma^2\left(\frac{1}{2}+d\lambda^2\right)\\
\end{align*}
We take the expectation of this over the distribution of $\theta_1$ to
get that the expected loss is approximately
$$\frac{\sigma^2}{2}\left({\mathbb
  E}(\pi(x))-2d{\mathbb E}(\pi(x)\lambda(x)^2)\right)$$
We have assumed that $d$ is small with respect to this second term, so
the expected loss is approximately $\frac{\sigma^2}{2}{\mathbb
  E}(\pi(x))$.

(ii) Since $\sigma_1$ and $\sigma_2$ are small, we can assume that
$\lambda(x_1)\approx\lambda(x_2)$. We will let $\lambda$ denote this
common value. Suppose that $x_1$ is the observed value of $X_1$ and
$x_2$ is the observed value of $X_2$, and that
$x_1-x_2>\lambda(\sigma_1{}^2-\sigma_2{}^2)$, but
$x_1-x_2<\hat{\lambda}(\sigma_1{}^2-\sigma_2{}^2)$ so that $x_1$ and
$x_2$ are mis-ranked compared to the ranking under the true prior. Let
$\theta$ be the underlying parameter value for $x_1$, and let $\phi$
be the underlying parameter value for $x_2$. The expected increase in
the loss function due to this misranking, compared to using the
true prior, is then
\begin{align*}
\iint
\pi_{x_1}(\theta)\pi_{x_2}(\phi)(\theta-\phi)\,d\theta\,d\phi&=\int_{-\infty}^\infty
\pi_{x_1}(\theta)\theta\,d\theta-\int_{-\infty}^\infty\pi_{x_2}(\phi)\phi\,d\phi\\
&=x_1-\lambda\sigma_1{}^2-(x_2-\lambda\sigma_2{}^2)\\
&=x_1-x_2-\lambda(\sigma_1{}^2-\sigma_2{}^2)\\
\end{align*}

(where $\pi_{x_1}(\theta)$ and $\pi_{x_2}(\phi)$ are the posterior
distributions of $\theta$ and $\phi$ given observations $x_1$ and
$x_2$ respectively, under the true prior). Similarly, if
$\hat{\lambda}(\sigma_1{}^2-\sigma_2{}^2)<x_1-x_2<\lambda(\sigma_1{}^2-\sigma_2{}^2)$,
then the expected increase in loss is
$x_2-x_1-\lambda(\sigma_2{}^2-\sigma_1{}^2)$.

Now suppose we fix $X_1=x_1$ and
take the expectation of the loss over $X_2$. The expected loss due to
mis-ranking them is
\begin{equation}\label{EqExpLossFixedX1}
\int_{x_1-\hat{\lambda}(\sigma_1{}^2-\sigma_2{}^2)}^{x_1-\lambda(\sigma_1{}^2-\sigma_2{}^2)}
\pi_2(x_2)(x_1-x_2-\lambda(\sigma_1{}^2-\sigma_2{}^2))\,dx_2\end{equation} where
$\pi_2$ is the marginal density of $x_2$.
In the case where
$(\hat{\lambda}-\lambda)(\sigma_1{}^2-\sigma_2{}^2)<0$, we get this by
calculating expected misranking loss over all values for which
$x_1$ and $x_2$ are misranked (compared to posterior mean ranking
using the true prior). In the case where
$(\hat{\lambda}-\lambda)(\sigma_1{}^2-\sigma_2{}^2)>0$, calculating
the expected misranking loss over all values where $x_1$ and $x_2$ are
misranked gives
$$\int^{x_1-\hat{\lambda}(\sigma_1{}^2-\sigma_2{}^2)}_{x_1-\lambda(\sigma_1{}^2-\sigma_2{}^2)}
\pi_2(x_2)(x_2-x_1-\lambda(\sigma_2{}^2-\sigma_1{}^2))\,dx_2$$ and by
reversing the limits and negating the integrand, we get the formula
from Equation~\ref{EqExpLossFixedX1} in this case also.

Since $\sigma_1{}^2$ and $\sigma_2{}^2$ are both small, we can assume
that $\pi$, $\pi_1$ and $\pi_2$ are approximately constant arround
$x_1$, so that $\pi_2(x_2)\approx\pi_1(x_1)\approx\pi(x_1)$ for all
$x_2$ in the relevant range. The integral is then approximately
\begin{align*}
\pi(x_1)\int_{x_1-\lambda(\sigma_1{}^2-\sigma_2{}^2)}^{x_1-\hat{\lambda}(\sigma_1{}^2-\sigma_2{}^2)}
\left(x_2-(x_1-\lambda(\sigma_1{}^2-\sigma_2{}^2))\right)\,dx_2&=\frac{1}{2}\pi(x_1)\left((x_1-\hat{\lambda}(\sigma_1{}^2-\sigma_2{}^2))-(x_1-\lambda(\sigma_1{}^2-\sigma_2{}^2))\right)^2\\
&=\frac{1}{2}\pi(x_1)(\lambda-\hat{\lambda})^2(\sigma_1{}^2-\sigma_2{}^2)^2\\
\end{align*}

For the overall mis-ranking loss, we take the expectation of this over
$x_1$. We are usually particularly interested in the mis-ranking loss
of the upper tail, that is the expected loss due to all misrankings in
the upper tail, so we usually take the expectation over the
distribution of $x_1$ for values $x_1>a$ for some chosen $a$. This is
given by

$$\frac{1}{2}(\sigma_1{}^2-\sigma_2{}^2)^2\int_{a}^\infty \pi(x_1)^2(\lambda(x_1)-\hat{\lambda}(x_1))^2\,dx_1$$

(iii) we calculate this loss by substituting $\hat{\lambda}=0$
into our expression for the additional loss, we get 
$$\frac{1}{2}\left(\sigma_1{}^2-\sigma_2{}^2\right)\int_a^\infty \pi(x)^2\lambda(\theta)^2\,d\theta=\frac{1}{2}\left(\sigma_1{}^2-\sigma_2{}^2\right)\int_a^\infty \pi'(x)^2\,d\theta$$

\end{proof}



%
%
%
%

\section{Non-Parametric Prior}
\label{NPPriorProof}

\begin{lemma}
Let $\pi$ be a discrete distribution with probability at least
$\frac{1}{r+1}$ in the interval $[x-a,x+a]$ for some $a>0$. Let
$\hat{\theta}$ be the posterior mean for an observation $x$ with
standard error $\sigma$. Then
$$|\hat{\theta}-x|\leqslant a+\sigma\sqrt{2\log(r)}$$
\end{lemma}

\begin{proof}
Let the support of $\pi$ be the values $x+b_i$, with probabilities
$\pi_i$. Then for the posterior distribution of $\theta$, the
probability of $x+b_i$ is 
$$\frac{\pi_ie^{-\frac{b_i{}^2}{2\sigma^2}}}{\sum
  \pi_je^{-\frac{b_j{}^2}{2\sigma^2}}}$$
The posterior mean is therefore 
$$\hat{\theta}=x+\frac{\sum b_i\pi_ie^{-\frac{b_i{}^2}{2\sigma^2}}}{\sum
  \pi_je^{-\frac{b_j{}^2}{2\sigma^2}}}$$

We see that the difference $|\hat{\theta}-x|$ is maximised when the
$b_i$ all have the same sign, which we will w.l.o.g. assume to be
positive. It is clear that $\hat{\theta}$ is maximised by setting all
the $b_i$ in the interval $[0,a]$ to equal $a$, since this both
minimises the posterior probability of the interval $[x,x+a]$ and
maximises the posterior mean conditional on lying in this interval.
We will therefore assume that $b_1=a$, and $\pi_1=\frac{1}{r+1}$, then we have
$$|\hat{\theta}-x|=\frac{\sum b_i\pi_ie^{-\frac{b_i{}^2}{2\sigma^2}}}{\sum
  \pi_je^{-\frac{b_j{}^2}{2\sigma^2}}}
=\frac{
 \frac{ae^{-\frac{a^2}{2\sigma^2}}}{r+1}+\sum_{i=2}^k
  b_i\pi_ie^{-\frac{b_i{}^2}{2\sigma^2}}}{\sum_{i=1}^k
  b_i\pi_ie^{-\frac{b_i{}^2}{2\sigma^2}}}
=\frac{
ae^{-\frac{a^2}{2\sigma^2}}}{(r+1)\sum \pi_je^{-\frac{b_j{}^2}{2\sigma^2}}}+\frac{\sum_{i=m+1}^k
  b_i\pi_ie^{-\frac{b_i{}^2}{2\sigma^2}}}{\sum
  \pi_je^{-\frac{b_j{}^2}{2\sigma^2}}}
$$

For fixed $b_i$, and fixed $\sum_{i=2}^k \pi_i$, if 
$\frac{
ae^{-\frac{a^2}{2\sigma^2}}}{(r+1)\sum \pi_je^{-\frac{b_j{}^2}{2\sigma^2}}}+\frac{\sum_{i=2}^k
  b_i\pi_ie^{-\frac{b_i{}^2}{2\sigma^2}}}{\sum
  \pi_je^{-\frac{b_j{}^2}{2\sigma^2}}}=C$ then
$\frac{
ae^{-\frac{a^2}{2\sigma^2}}}{(r+1)}+\sum_{i=2}^k
  b_i\pi_ie^{-\frac{b_i{}^2}{2\sigma^2}}=C\sum
  \pi_je^{-\frac{b_j{}^2}{2\sigma^2}}$
$$\frac{ (a-C)}{(r+1)}e^{-\frac{a^2}{2\sigma^2}}+\sum_{i=2}^k
  (b_i-C)\pi_ie^{-\frac{b_i{}^2}{2\sigma^2}}=0$$ This gives that the
  contours are linear functions in $\pi$, so the maximum value of $C$
  occurs at a vertex with only one non-zero value of
  $\pi_i$. The value of $\hat{\theta}-x$ is maximised
  subject to $b_1=a$, $\pi_1=\frac{1}{r+1}$ by setting $k=2$,
  $\pi_2=\frac{r}{1+r}$ and choosing the value of $b_2$ to maximise
  the resulting quantity. In this case, we have
$$|\hat{\theta}-x|=\frac{ae^{-\frac{a^2}{2\sigma^2}}+rbe^{-\frac{b^2}{2\sigma^2}}}{e^{-\frac{a^2}{2\sigma^2}}+re^{-\frac{b^2}{2\sigma^2}}}=a+\frac{r(b-a)e^{-\frac{b^2}{2\sigma^2}}}{e^{-\frac{a^2}{2\sigma^2}}+re^{-\frac{b^2}{2\sigma^2}}}
=a+\frac{r(b-a)}{r+e^{\frac{b^2-a^2}{2\sigma^2}}}
$$
Substituting $a=v\sigma$, $b-a=w\sigma$, this expression becomes
$$\hat{\theta}-x=a+\sigma\frac{rw}{r+e^{\frac{w^2+2vw}{2}}}
$$
The derivative of $\frac{r+e^{\frac{w^2+2vw}{2}}}{rw}$ with respect to
$w$ is 
$\left(\frac{w+v}{rw}-\frac{1}{rw^2}\right)e^{\frac{w^2+2vw}{2}}-\frac{1}{w^2}$. We see that for
$w>\sqrt{2\log(r)}$, $w>\sqrt{2}$ and $v\geqslant 0$, we have 
$$\left(\frac{w+v}{rw}-\frac{1}{rw^2}\right)e^{\frac{w^2+2vw}{2}}-\frac{1}{w^2}\geqslant
\left(\frac{w+v}{rw}-\frac{1}{rw^2}\right)e^{\log(r)}e^{vw}-1\geqslant
1+\frac{v}{w}-\frac{2}{w^2}\geqslant 0$$ so
$\frac{r+e^{\frac{w^2+2vw}{2}}}{rw}$ is increasing. Therefore
$\frac{rw}{r+e^{\frac{w^2+2vw}{2}}}$ is decreasing. This means that
assuming $r>e$, we have that 
$\frac{rw}{r+e^{\frac{w^2+2vw}{2}}}\leqslant
\frac{r\sqrt{2\log(r)}}{r(1+e^{v\sqrt{2\log(r)}})}\leqslant\sqrt{2\log(r)}$
for all $w>\sqrt{2\log(r)}$. Meanwhile, we always have
$\frac{rw}{r+e^{\frac{w^2+2vw}{2}}}\leqslant w$, so we always have
$\frac{rw}{r+e^{\frac{w^2+2vw}{2}}}\leqslant \sqrt{2\log(r)}$, and
therefore 
$$|\hat{\theta}-x|\leqslant |a|+\sigma\sqrt{2\log(r)}$$ 
\end{proof}

\begin{lemma}
For a sample of $n$ datapoints and their corresponding standard
errors, the MLE estimate for the prior distribution always assigns
probability at least $\frac{1-e^{-\frac{1}{2}}}{n}$ to the interval
$(x-\sigma\sqrt{2\log(n)+1},x+\sigma\sqrt{2\log(n)+1})$, for every
observed data point $(x,\sigma)$.
\end{lemma}

\begin{proof}
Suppose the MLE assigns probability $\pi_i$ to point $b_i$. We will
separate the points $b_i$ into points that are in the interval
$I=(x-\sigma\sqrt{2\log(n)+1},x+\sigma\sqrt{2\log(n)+1})$,
and points that are not. Suppose the first $m$ points are in the
interval $I$ and the remaining points are not. We are aiming to show
that $\sum_{i=1}^m \pi_i>\frac{1}{n}$. Suppose this is not the
case. We will then show that the distribution assigning probability
$\pi_i$ to each point $b_i$ is not the MLE by constructing a prior
distribution with larger likelihood.  Let $\phi=\sum_{i=1}^m
\pi_i$. Let $\mathcal X$ be the data set, and let
$(x,\sigma)\in{\mathcal X}$ be a data point. The log-likelihood of the
data can be represented as $l({\mathcal X}\setminus
(x,\sigma))+l(x,\sigma)$, i.e. as the likelihood of the point
$(x,\sigma)$ plus the likelihood of the remainder of the data
points. For a data point $(y,\sigma_y)$, we will use $L_I(y,\sigma_y)$
to represent the conditional likelihood of $y$ given that its
corresponding value of $\theta$ is contained in $I$, and
$L_{\overline{I}}(y)$ for the conditional likelihood of $y$ given that its
corresponding value of $\theta$ is not contained in $I$. We have that
the likelihood of $y$ is $\phi
L_{I}(y)+(1-\phi)L_{\overline{I}}(y)$. If we change the prior to have
probability $\alpha$ at $x$ and $1-\alpha$ times the previous prior,
then the log-likelihood is larger than $l({\mathcal X}\setminus
(x,\sigma))+(n-1)\log(1-\alpha)+\log\left((1-\alpha)L(x)+\frac{\alpha}{\sqrt{2\pi}\sigma}\right)$
The increase in log-likelihood is therefore
$$(n-1)\log(1-\alpha)+\log\left(\frac{(1-\alpha)L(x)+\frac{\alpha}{\sqrt{2\pi}\sigma}}{L(x)}\right)=(n-1)\log(1-\alpha)+\log\left(1-\alpha+\frac{\alpha}{\sqrt{2\pi}\sigma
  L(x)}\right)$$
For this to be an increase, we need 
\begin{align*}
(n-1)\log(1-\alpha)+\log\left(1-\alpha+\frac{\alpha}{\sqrt{2\pi}\sigma
    L(x)}\right)&>0\\
(1-n)\log(1-\alpha)&<\log\left(1-\alpha+\frac{\alpha}{\sqrt{2\pi}\sigma
    L(x)}\right)\\
(1-\alpha)^{(1-n)}&<1-\alpha+\frac{\alpha}{\sqrt{2\pi}\sigma L(x)}\\
(1-\alpha)\left((1-\alpha)^{-n}-1\right)&<\frac{\alpha}{\sqrt{2\pi}\sigma
    L(x)}\\
\end{align*}
If we substitute $\alpha=\frac{\beta}{n}$, where $\beta<1$, then
$(1-\alpha)^{-n}\approx e^{\beta}$ for large $n$. We therefore need 
\begin{align*}
(1-\alpha)\left(e^{\beta}-1\right)&<\frac{\alpha}{\sqrt{2\pi}\sigma L(x)}\\
{\sqrt{2\pi}\sigma L(x)}&<\frac{\alpha}{(1-\alpha)\left(e^{\beta}-1\right)}\\
{\sqrt{2\pi}\sigma L(x)}&<\frac{\beta}{(n-\beta)\left(e^{\beta}-1\right)}\\
\end{align*}
Now we know that $L(x)=L_I(x)+L_{\overline I}(x)$, and
$L_I(x)\leqslant \frac{\phi}{\sqrt{2\pi}\sigma}$, since the prior
probability of the interval $I$ is at most $\phi$, and for each point
$b_i$, the likelihood is
$\frac{\pi_ie^{-\frac{(x-b_i)^2}{2\sigma^2}}}{\left(\sum_{i=1}^m
  \pi_i\right)\sqrt{2\pi}\sigma}$ so $L_{I}(x)=\sum_{i=1}^m\frac{\pi_ie^{-\frac{(x-b_i)^2}{2\sigma^2}}}{\left(\sum_{i=1}^m
  \pi_i\right)\sqrt{2\pi}\sigma}<\frac{1}{\sqrt{2\pi}\sigma}$. Also
$L_{\overline{I}}(x)\leqslant
(1-\phi)\frac{e^{-\left(\log(n)+\frac{1}{2}\right)}}{\sqrt{2\pi}\sigma}=(1-\phi)\frac{e^-{\frac{1}{2}}}{n\sqrt{2\pi}\sigma}$
so $L(x)\leqslant \frac{\phi+\frac{(1-\phi)}{n} e^{-\frac{1}{2}}}{\sqrt{2\pi}\sigma}$.
Therefore, provided that 
$$\frac{\beta}{(n-\beta)\left(e^{\beta}-1\right)}>\phi+\frac{(1-\phi)}{n}
e^{-\frac{1}{2}}$$
we will have an improvement in likelihood. we see that as
$\beta\rightarrow0$, $\frac{\beta}{e^\beta-1}\rightarrow 1$.
For small $\beta$, the left-hand side is approximately $\frac{1}{n}$. 
For the right-hand side, we are given that $\phi<\frac{1-e^{-\frac{1}{2}}}{n}$, so the
right-hand side is less than
$$\frac{1-e^{-\frac{1}{2}}}{n}+\frac{e^{-\frac{1}{2}}}{n}=\frac{1}{n}$$
so the required inequality holds.

\end{proof}

\section{Optimal Parameter Values and Loss Functions for Simulations}
\label{OptimalParameters}

Recall that for a distribution with density
function $\pi(\theta)$, we define
$\lambda(\theta)=-\frac{\pi'(\theta)}{\pi(\theta)}$, and that the best
choice of estimating prior to use for ranking is chosen to minimise
$$\int_a^\infty \pi(\theta)^2(\lambda(\theta)-\hat{\lambda}(\theta))^2\,d\theta$$
where $\lambda(\theta)$ and $\pi(\theta)$ are for the true
prior, while $\hat{\lambda}(\theta)$ is for the estimating
prior. We evaluate this loss function for each combination of
priors. 

\subsection{Loss functions}

\subsubsection{Normal Estimated by Normal}

If the true prior is normal with mean 0 variance $\tau^2$, and the
estimated prior has mean 0, variance $\hat{\tau}^2$, then the loss
function is given by
\begin{align*}
\int_a^\infty
\pi(\theta)^2\left(\frac{\theta}{\tau^2}-\frac{\theta}{\hat{\tau}^2}\right)^2\,d\theta&=\left(\frac{1}{\tau^2}-\frac{1}{\hat{\tau}^2}\right)^2\int_a^\infty
\theta^2\frac{e^{-\frac{\theta^2}{\tau^2}}}{2\pi\tau^2}\,d\theta\\
&=\left(\frac{1}{\tau^2}-\frac{1}{\hat{\tau}^2}\right)^2\left(\frac{\left[-\theta
e^{-\frac{\theta^2}{\tau^2}}  \right]_a^\infty}{4\pi}+\int_a^\infty
\frac{e^{-\frac{\theta^2}{\tau^2}}}{4\pi}\,d\theta\right)\\
&=\left(\frac{1}{\tau^2}-\frac{1}{\hat{\tau}^2}\right)^2\left(\frac{a
e^{-\frac{a^2}{\tau^2}} }{4\pi}+\frac{\tau}{4\sqrt{\pi}}\left(1-\Phi\left(\frac{\sqrt{2}a}{\tau}\right)\right)\right)\\
\end{align*}

\subsubsection{Exponential estimated by normal}

For the exponential true prior we have $\pi(\theta)=\lambda
e^{-\lambda\theta}$ and $\lambda(\theta)=\lambda$. Meanwhile, for the
normal estimating prior, we have that
$\hat{\lambda}(\hat{\theta})=\frac{\hat{\theta}}{\tau^2}$. We are
aiming to choose $\hat{\tau}$ so as to minimise
\begin{align*}
\int_a^\infty
\pi(\theta)^2\left(\lambda-\hat{\lambda}(\theta)\right)^2\,d\theta
&=\int_a^\infty
\pi(\theta)^2\left(\lambda-\frac{\theta}{\hat{\tau}^2}\right)^2\,d\theta\\
&=\int_a^\infty
\lambda^2e^{-2\lambda\theta}\left(\lambda-\frac{\theta}{\hat{\tau}^2}\right)^2\,d\theta\\
\end{align*}

We recall that 
\begin{align*}
\int_a^\infty e^{-2\lambda\theta}\,d\theta&=\frac{e^{-2\lambda a}}{2\lambda}\\
\int_a^\infty \theta e^{-2\lambda\theta}\,d\theta&=\frac{e^{-2\lambda
    a}(2\lambda a+1)}{4\lambda^2}\\
\int_a^\infty \theta^2e^{-2\lambda\theta}\,d\theta&=\frac{e^{-2\lambda
    a}(2\lambda^2a^2+2\lambda a+1)}{4\lambda^3}\\
\end{align*}

Therefore, the objective function is 
$$\frac{e^{-2\lambda a}}{\lambda}\left(\frac{\lambda^4}{2}-\frac{2\lambda^3
  a+\lambda^2}{2\hat{\tau}^2}+\frac{2\lambda^2a^2+2\lambda a+1}{4\hat{\tau}^4}\right)$$

\subsubsection{Pareto estimated by normal}

For the normal estimating prior, we have
$\hat{\lambda}(\theta)=\frac{\theta}{\hat{\tau}^2}$. For the Pareto
true prior, we have $\lambda(\theta)=\frac{\alpha+1}{\theta}$. The
objective function is therefore 
\begin{align*}
\int_a^\infty
\pi(\theta)^2(\lambda(\theta)-\hat{\lambda}(\theta))^2\,d\theta&=\alpha^2\int_a^\infty
\frac{\eta^{2\alpha}}{\theta^{2\alpha+2}}\left(\frac{\alpha+1}{\theta}-\frac{\theta}{\hat{\tau}^2}\right)^2\,d\theta\\
&=\alpha^2\eta^{2\alpha}\int_a^\infty \left(\frac{(\alpha+1)^2}{\theta^{2\alpha+4}}-\frac{2(\alpha+1)}{\theta^{2\alpha+2}\hat{\tau}^2}+\frac{1}{\theta^{2\alpha}\hat{\tau}^4}\right)\,d\theta\\
&=\alpha^2\eta^{2\alpha}\left[ -\frac{(\alpha+1)^2}{(2\alpha+3)\theta^{2\alpha+3}}+\frac{2(\alpha+1)}{(2\alpha+1)\theta^{2\alpha+1}\hat{\tau}^2}-\frac{1}{(2\alpha-1)\theta^{2\alpha-1}\hat{\tau}^4}\right]_a^\infty\\
&=\frac{\alpha^2\eta^{2\alpha}}{a^{2\alpha+3}}\left(\frac{(\alpha+1)^2}{(2\alpha+3)}-\frac{2(\alpha+1)a^2}{(2\alpha+1)\hat{\tau}^2}+\frac{a^4}{(2\alpha-1)\hat{\tau}^4}\right)\\
\end{align*}

\subsubsection{Normal estimated by exponential}

The expected loss is 
\begin{align*}
\hat{\lambda}^2\int_a^\infty
\pi(\theta)^2\,d\theta+2\hat{\lambda}\int_a^\infty
\pi(\theta)\pi'(\theta)\,d\theta+\int_a^\infty \pi'(\theta)^2\,d\theta&=\int_a^\infty \pi'(\theta)^2\,d\theta-\hat{\lambda}
\pi(a)^2+\hat{\lambda}^2\int_a^\infty
\pi(\theta)^2\,d\theta\\
&=\int_a^\infty
\frac{\theta^2}{2\pi\tau^6}e^{-\frac{\theta^2}{\tau^2}}\,d\theta-
\hat{\lambda} \frac{e^{-\frac{a^2}{\tau^2}}}{2\pi\tau^2}+
\hat{\lambda}^2\int_a^\infty \frac{e^{-\frac{\theta^2}{\tau^2}}}{2\pi\tau^2}\,d\theta\\
\end{align*}

We have 
\begin{align*}
\int_a^\infty
\theta\frac{2\theta}{\tau^2}e^{-\frac{\theta^2}{\tau^2}}\,d\theta&=\left[-\theta
  e^{-\frac{\theta^2}{\tau^2}}\right]_a^\infty+\int_a^\infty
e^{-\frac{\theta^2}{\tau^2}}\,d\theta\\
&=ae^{-\frac{a^2}{\tau^2}}+\sqrt{\pi}\tau\left(1-\Phi\left(\frac{\sqrt{2}a}{\tau}\right)\right)\\
\int_a^\infty \frac{e^{-\frac{\theta^2}{\tau^2}}}{2\pi\tau^2}\,d\theta&=\frac{1-\Phi\left(\frac{\sqrt{2}a}{\tau}\right)}{2\sqrt{\pi}\tau}\\
\end{align*}
so the expected loss is
$$\frac{\hat{\lambda}^2}{2\sqrt{\pi}\tau}\left(1-\Phi\left(\frac{\sqrt{2}a}{\tau}\right)\right)
-\frac{\hat{\lambda}e^{-\frac{a^2}{\tau^2}}}{2\pi\tau^2}+\frac{ae^{-\frac{a^2}{\tau^2}}}{4\pi\tau^4}+\frac{1-\Phi\left(\frac{\sqrt{2}a}{\tau}\right)}{4\sqrt{\pi}\tau^3}$$

\subsubsection{Exponential Estimated by Exponential}

If the true prior is exponential with rate $\lambda$, and the
estimated prior is exponential with rate $\hat{\lambda}$, then the
loss function is given by
\begin{align*}
\int_a^\infty
\pi(\theta)^2\left(\lambda-\hat{\lambda}\right)^2\,d\theta&=\left(\lambda-\hat{\lambda}\right)^2\int_a^\infty
    \lambda^2e^{-2\lambda\theta}\,d\theta\\
&=\left(\lambda-\hat{\lambda}\right)^2\left[\frac{-\lambda e^{-2\lambda\theta}}{2}\right]_a^\infty\\
&=\frac{\lambda\left(\lambda-\hat{\lambda}\right)^2}{2}e^{-2\lambda a}\\
\end{align*}

\subsubsection{Pareto estimated by exponential}

The loss function is 
\begin{align*}
\hat{\lambda}^2\int_a^\infty
\left(\frac{\alpha\eta^\alpha}{\theta^{\alpha+1}}\right)^2\,d\theta-2\hat{\lambda}\int_a^\infty
\left(\frac{\alpha\eta^\alpha}{\theta^{\alpha+1}}\right)\left(\frac{\alpha(\alpha+1)\eta^\alpha}{\theta^{\alpha+2}}\right)\,d\theta+\int_a^\infty
\left(\frac{\alpha(\alpha+1)\eta^\alpha}{\theta^{\alpha+2}}\right)^2\,d\theta&=\frac{\hat{\lambda}^2\alpha^2\eta^{2\alpha}}{(2\alpha+1)a^{2\alpha+1}}-2\frac{\hat{\lambda}\alpha^2(\alpha+1)\eta^{2\alpha}}{(2\alpha+2)a^{2\alpha+2}}+\frac{\alpha^2(\alpha+1)^2\eta^{2\alpha}}{(2\alpha+3)a^{2\alpha+3}}\\
&=\frac{\alpha^2\eta^{2\alpha}}{a^{2\alpha+1}}\left(\frac{\hat{\lambda}^2}{(2\alpha+1)}-\frac{\hat{\lambda}}{a}+\frac{(\alpha+1)^2}{(2\alpha+3)a^2}\right)
\end{align*}

\subsubsection{Normal Estimated by Pareto}

For the Normal estimated by Pareto, we have
$\hat{\lambda}(\theta)=\frac{\hat{\alpha}+1}{\theta}$. The loss function is therefore
\begin{align*}
\frac{1}{2\pi\tau^2}\int_a^\infty
e^{-\frac{\theta^2}{\tau^2}}\left(\frac{\theta}{\tau^2}-\frac{\hat{\alpha}+1}{\theta}\right)^2\,d\theta&=\frac{1}{2\pi\tau^2}\int_a^\infty
e^{-\frac{\theta^2}{\tau^2}}\left(\frac{\theta^2}{\tau^4}-\frac{2(\hat{\alpha}+1)}{\tau^2}+\frac{(\hat{\alpha}+1)^2}{\theta^2}\right)\,d\theta
\\
&=\frac{1}{2\pi\tau^2}\left((\hat{\alpha}+1)^2\int_a^\infty
\frac{1}{\theta^2}e^{-\frac{\theta^2}{\tau^2}}\,d\theta-2(\hat{\alpha}+1)\int_a^\infty
\frac{1}{\tau^2}e^{-\frac{\theta^2}{\tau^2}}\,d\theta +\int_a^\infty \frac{\theta^2}{\tau^4}e^{-\frac{\theta^2}{\tau^2}}\,d\theta\right)\\
&=\frac{1}{2\pi\tau^2}\left((\hat{\alpha}+1)^2\int_a^\infty
\frac{1}{\theta^2}e^{-\frac{\theta^2}{\tau^2}}\,d\theta-2(\hat{\alpha}+1)\int_a^\infty
\frac{1}{\tau^2}e^{-\frac{\theta^2}{\tau^2}}\,d\theta +\int_a^\infty
\frac{\theta}{\tau^4}\theta e^{-\frac{\theta^2}{\tau^2}}\,d\theta\right)\\
&=\frac{1}{2\pi\tau^2}\left((\hat{\alpha}+1)^2\int_a^\infty
\frac{1}{\theta^2}e^{-\frac{\theta^2}{\tau^2}}\,d\theta-2(\hat{\alpha}+1)\int_a^\infty
\frac{1}{\tau^2}e^{-\frac{\theta^2}{\tau^2}}\,d\theta +\left[-
\frac{\theta}{2\tau^2} e^{-\frac{\theta^2}{\tau^2}}\right]_a^\infty+\int_a^\infty\frac{1}{2\tau^2} e^{-\frac{\theta^2}{\tau^2}}\,d\theta\right)\\
&=\frac{1}{2\pi\tau^2}\left((\hat{\alpha}+1)^2\int_a^\infty
\frac{1}{\theta^2}e^{-\frac{\theta^2}{\tau^2}}\,d\theta-\left(2\hat{\alpha}+\frac{3}{2}\right)\frac{\sqrt{\pi}}{\tau}\left(1-\Phi\left(\frac{\sqrt{2}a}{\tau}\right)\right) +
\frac{a}{2\tau^2} e^{-\frac{a^2}{\tau^2}}\right)\\
\end{align*}

\subsubsection{Exponential estimated by Pareto}

For the Exponential estimated by Pareto, we have
$\hat{\lambda}(\theta)=\frac{\hat{\alpha}+1}{\theta}$. The loss is therefore
\begin{align*}
\lambda^2\int_a^\infty
e^{-2\lambda\theta}\left(\lambda-\frac{\hat{\alpha}+1}{\theta}\right)^2\,d\theta&=\lambda^2\int_a^\infty
e^{-2\lambda\theta}\left(\lambda^2-\frac{2\lambda(\hat{\alpha}+1)}{\theta}+\frac{(\hat{\alpha}+1)^2}{\theta^2}\right)\,d\theta
\\
&=\lambda^2\left( (\hat{\alpha}+1)^2\int_a^\infty
\frac{1}{\theta^2}e^{-2\lambda\theta}\,d\theta-2\lambda(\hat{\alpha}+1)\int_a^\infty
\frac{1}{\theta}e^{-2\lambda\theta}\,d\theta +\lambda^2\int_a^\infty e^{-2\lambda\theta}\,d\theta\right)\\
\end{align*}

\subsubsection{Pareto Estimated by Pareto}

If the true prior is Pareto with minimum $\eta$ and index $\alpha$, and the
estimated prior is Pareto with minimum $\eta$ and index $\hat{\alpha}$, then the
loss function is given by
\begin{align*}
\int_a^\infty
\pi(\theta)^2\left(\frac{\alpha+1}{\theta}-\frac{\hat{\alpha}+1}{\theta}\right)^2\,d\theta&=\left(\alpha-\hat{\alpha}\right)^2\int_a^\infty
   \frac{\alpha^2\eta^{2\alpha}}{\theta^{2\alpha+4}}\,d\theta\\
&=\left(\alpha-\hat{\alpha}\right)^2\left[-\frac{\alpha^2\eta^{2\alpha}}{(2\alpha+3)\theta^{2\alpha+3}}\right]_a^\infty\\
&=\left(\alpha-\hat{\alpha}\right)^2\frac{\alpha^2\eta^{2\alpha}}{(2\alpha+3)a^{2\alpha+3}}\\
\end{align*}

\subsubsection{MLE Ranking with Normal Prior}

For a normal true prior, we have that the loss from using the MLE ranking is 
\begin{align*}
\int_a^\infty \pi'(\theta)^2\,d\theta&=\int_a^\infty
\frac{\theta^2}{2\pi\tau^6}e^{-\frac{\theta^2}{\tau^2}}\,d\theta=\frac{1}{4\pi\tau^4}\left(\left[-\theta
  e^{-\frac{\theta^2}{\tau^2}}\right]_a^\infty+\int_a^\infty e^{-\frac{\theta^2}{\tau^2}}\,d\theta\right)=\frac{1}{4\pi\tau^4}\left(a
  e^{-\frac{a^2}{\tau^2}}+\sqrt{\pi}\tau\left(1-\Phi\left(\frac{\sqrt{2}a}{\tau}\right)\right)\right)\\
\end{align*}

For $\tau=1$, $a=1.281552$ this loss is 0.02466714.

\subsubsection{MLE Ranking with Exponential Prior}

The expected loss function using the MLE ranking is 
\begin{align*}
\int_a^\infty \pi'(\theta)^2\,d\theta&=\int_a^\infty
\left(-\lambda^2e^{-\lambda\theta}\right)^2\,d\theta=\lambda^4\int_a^\infty
e^{-2\lambda\theta}\,d\theta=\frac{\lambda^3}{2}\left[-e^{-2\lambda\theta}\right]_a^\infty=\frac{\lambda^3e^{-2\lambda a}}{2}\\
\end{align*}

Substituting $\lambda=1$ and $a=\log(10)$ this gives 0.005.

\subsubsection{MLE Ranking with Pareto Prior}

For the Pareto true prior, the expected additional loss from using
the MLE ranking is 
\begin{align*}
\int_a^\infty \pi'(\theta)^2\,d\theta&=\int_a^\infty
\left(\frac{\alpha(\alpha+1)\eta^\alpha}{\theta^{\alpha+2}}\right)^2\,d\theta=\alpha^2(\alpha+1)^2\eta^{2\alpha}\int_a^\infty \theta^{-(2\alpha+4)}\,d\theta=\frac{\alpha^2(\alpha+1)^2\eta^{2\alpha}}{2\alpha+3}\left[-\theta^{-(2\alpha+3)}\right]_a^\infty=\frac{\alpha^2(\alpha+1)^2\eta^{2\alpha}}{(2\alpha+3)a^{2\alpha+3}}\\
\end{align*}
Substituting  $\alpha=2$, $\eta=\frac{1}{2}$ and
$a=\frac{\sqrt{10}}{2}$, we get the loss is 
\begin{align*}
\frac{2^23^2\left(\frac{1}{2}\right)^4}{7\left(\frac{\sqrt{10}}{2}\right)^7}&=\frac{288}{7000\sqrt{10}}=0.01301051\\
\end{align*}

\subsection{Optimal Parameter estimates}

\subsubsection{Exponential estimated by normal}

The loss function is minimised by 
\begin{align*}
\frac{1}{\tau^2}&=\frac{2\lambda^3a+\lambda^2}{2\lambda^2a^2+2\lambda
  a+1}\\
\tau^2&=\frac{1}{\lambda^2}\left(\frac{2\lambda^2a^2+2\lambda
  a+1}{2\lambda a+1}\right)
\end{align*}

Substituting $\lambda=1$ and $a=\log(10)$ (the 90th percentile of the
exponential distribution) gives 
$$\hat{\tau}=\sqrt{\frac{2\log(10)^2+2\log(10)+1}{2\log(10)+1}}=1.700526$$
and the expected loss is
$$0.01\left(\frac{1}{2}-\frac{(2\log(10)+1)^2}{4(2\log(10)^2+2\log(10)+1)}\right)=0.0001542356$$


\subsubsection{Pareto estimated by normal}

The loss function is minimised by 
\begin{align*}
\frac{1}{\tau^2}&=\frac{\left(\frac{\alpha+1}{(2\alpha+1)a^{2\alpha+1}}\right)}{\left(\frac{1}{(2\alpha-1)a^{2\alpha-1}}\right)}=\frac{(2\alpha-1)(\alpha+1)}{(2\alpha+1)a^2}\\
\tau^2&=\frac{2\alpha+1}{(\alpha+1)(2\alpha-1)}a^2\\
\end{align*}

For this value, the loss is 
\begin{align*}
\frac{4\alpha^2(\alpha+1)^2\eta^{2\alpha}}{(2\alpha+3)(2\alpha+1)^2a^{2\alpha+3}}
\end{align*}

Substituting the values $\alpha=2$, $\eta=\frac{1}{2}$ used in the
simulation and the corresponding 90th percentile
$a=\frac{\sqrt{10}}{2}$, we get that the optimal parameter $\tau$ has
\begin{align*}
\frac{1}{\tau^2}=\frac{(2\alpha-1)(\alpha+1)}{(2\alpha+1)a^2}=\frac{3\times 3}{5\times\frac{10}{4}}=0.72\\
\end{align*}
and the loss is 
\begin{align*}
\frac{4\alpha^2(\alpha+1)^2\eta^{2\alpha}}{(2\alpha+3)(2\alpha+1)^2a^{2\alpha+3}}&=
\frac{4\times2^2\times3^2\times\left(\frac{1}{2}\right)^{4}}{7\times5^2\times
  \left(\frac{\sqrt{10}}{2}\right)^{7}}
=0.002081682\\
\end{align*}


\subsubsection{Normal estimated by exponential}

If the true prior is normal, but we are using an exponential, then
recall that the best choice is 
$$\lambda=\frac{\pi(a)^2}{2\int_a^\infty \pi(\theta)^2\,d\theta}$$
We evaluate
\begin{align*}
\int_a^\infty
\pi(\theta)^2\,d\theta&=\frac{1}{2\pi\tau^2}\int_a^\infty
e^{-\frac{\theta^2}{\tau^2}}\,d\theta
=\frac{1}{2\sqrt{\pi}\tau}\int_a^\infty \frac{1}{\sqrt{2\pi}\frac{\tau}{\sqrt{2}}}e^{-\frac{\theta^2}{\tau^2}}\,d\theta
=\frac{1}{2\sqrt{\pi}\tau}\left(1-\Phi\left(\frac{\sqrt{2}a}{\tau}\right)\right)\\
\end{align*}

so the best choice of $\hat{\lambda}$ for the exponential estimating prior
is 
$$\hat{\lambda}=\frac{e^{-\frac{a^2}{\tau^2}}}{2\sqrt{\pi}\tau\left(1-\Phi\left(\frac{\sqrt{2}a}{\tau}\right)\right)}$$
For the simulation setting $\tau=1$, $a=1.281552$, this is 
$\hat{\lambda}=1.561386$
and the expected loss for our simulation is 0.000622064.


\subsubsection{Pareto estimated by exponential}

For the exponential prior, the best choice of $\lambda$ is given by 

$$\lambda=\frac{\pi(a)^2}{2\int_a^\infty \pi(\theta)^2\,d\theta}$$
We evaluate
\begin{align*}
\int_a^\infty
\pi(\theta)^2\,d\theta&=\alpha^2\int_a^\infty
\frac{\eta^{2\alpha}}{\theta^{2\alpha+2}}\,d\theta
=\alpha^2\left[\frac{-\eta^{2\alpha}}{(2\alpha+1)\theta^{2\alpha+1}}\right]_a^\infty
=\alpha^2\frac{\eta^{2\alpha}}{(2\alpha+1)a^{2\alpha+1}}\\
\end{align*}

so the best choice of $\hat{\lambda}$ for the exponential estimating prior
is 
$$\hat{\lambda}=\frac{\alpha^2\left(\frac{\eta^{2\alpha}}{a^{2\alpha+2}}\right)}{2\alpha^2\left(\frac{\eta^{2\alpha}}{(2\alpha+1)a^{2\alpha+1}}\right)}=\frac{2\alpha+1}{2a}$$
for this $\hat{\lambda}$ the expected loss is 
\begin{align*}
\frac{\alpha^2\eta^{2\alpha}}{4(2\alpha+3)a^{2\alpha+3}}\\
\end{align*}

Substituting the values $\alpha=2$, $\eta=\frac{1}{2}$ and the
corresponding 90th percentile $a=\frac{\sqrt{10}}{2}$, we get
$\hat{\lambda}=1.581139$ and the loss is
\begin{align*}
\frac{2^2\left(\frac{1}{2}\right)^{4}}{4\times 7\left(\frac{\sqrt{10}}{2}\right)^{7}}
&=\frac{8}{7000\sqrt{10}}=0.0003614032\\
\end{align*}


\subsubsection{Normal Estimated by Pareto}

For the Normal estimated by Pareto, the loss is minimised by 

\begin{align*}
\alpha+1&=\frac{\int_a^\infty
\frac{1}{\tau^2}e^{-\frac{\theta^2}{\tau^2}}\,d\theta}{\int_a^\infty
\frac{1}{\theta^2}e^{-\frac{\theta^2}{\tau^2}}\,d\theta}\\
\end{align*}

For this choice of $\alpha$, the loss is 
$$\frac{1}{2\pi\tau^2}\left(\int_a^\infty \frac{\theta^2}{\tau^4}e^{-\frac{\theta^2}{\tau^2}}\,d\theta-\frac{\left(\int_a^\infty \frac{1}{\tau^2}e^{-\frac{\theta^2}{\tau^2}}\,d\theta\right)^2}{\int_a^\infty \frac{1}{\theta^2}e^{-\frac{\theta^2}{\tau^2}}\,d\theta}\right)$$

For the case in our simulation, we have $\tau=1$ and $a=1.281552$. For
these values we calculate numerically 
\begin{align*}
\int_a^\infty \frac{1}{\theta^2}e^{-\frac{\theta^2}{\tau^2}}\,d\theta&=0.02706327\\
\int_a^\infty
\frac{\theta^2}{\tau^4}e^{-\frac{\theta^2}{\tau^2}}\,d\theta&=0.1549882\\
\int_a^\infty
\frac{1}{\tau^2}e^{-\frac{\theta^2}{\tau^2}}\,d\theta&=\sqrt{\pi}(1-\Phi(\sqrt{2}a))=0.06197059\\
\end{align*}

Substituting these into the formula, we get that the expected loss is 
$$\frac{1}{2\pi}\left(0.1549882-\frac{\left(0.06197059\right)^2}{0.02706327}\right)=0.002082605$$


\subsubsection{Exponential estimated by Pareto}

The loss is minimised by 

\begin{align*}
\alpha+1&=\frac{\lambda\int_a^\infty
\frac{1}{\theta}e^{-2\lambda\theta}\,d\theta}{\int_a^\infty
\frac{1}{\theta^2}e^{-2\lambda\theta}\,d\theta}\\
\end{align*}

Integrating by parts gives 
\begin{align*}
\int_a^\infty
\frac{1}{\theta}e^{-2\lambda\theta}\,d\theta&=\left[-\frac{1}{2\lambda\theta}e^{-2\lambda\theta}\right]_a^\infty-\int_a^\infty
\frac{1}{2\lambda\theta^2}e^{-2\lambda\theta}\,d\theta
=\frac{e^{-2\lambda a}}{2\lambda a}-\frac{1}{2\lambda}\int_a^\infty
\frac{1}{\theta^2}e^{-2\lambda\theta}\,d\theta\\
\end{align*}

We therefore get

\begin{align*}
\alpha+1&=\frac{\frac{e^{-2\lambda a}}{2 a}-\frac{1}{2}\int_a^\infty
\frac{1}{\theta^2}e^{-2\lambda\theta}\,d\theta}{\int_a^\infty
\frac{1}{\theta^2}e^{-2\lambda\theta}\,d\theta}=\frac{e^{-2\lambda a}}{2 a\int_a^\infty
\frac{1}{\theta^2}e^{-2\lambda\theta}\,d\theta}-\frac{1}{2}\\
\alpha&=\frac{e^{-2\lambda a}}{2 a\int_a^\infty
\frac{1}{\theta^2}e^{-2\lambda\theta}\,d\theta}-\frac{3}{2}\\
\end{align*}

We have  $\lambda=1$ and $a=\log(10)$, so numerically we obtain
\begin{align*}
\int_a^\infty
\frac{1}{\theta}e^{-2\lambda\theta}\,d\theta&=0.001829743\\
\int_a^\infty
\frac{1}{\theta^2}e^{-2\lambda\theta}\,d\theta&=0.0006834578\\
\int_a^\infty
e^{-2\lambda\theta}\,d\theta&=\frac{e^{-2\lambda a}}{2\lambda}=0.005\\
\end{align*}

This gives the optimal parameter estimate as
$$\hat\alpha=\frac{0.01}{2\times 0.0006834578\log(10)}-\frac{3}{2}=1.677186$$
so the expected loss is 
$0.0001014394$


\subsection{MLE Estimates for Parameters of Estimating Priors}
\label{AppMLEestderivation}

We will assume that $a$ is given for each simulation, and that our
objective is to estimate the parameters from the data for each
estimating prior so that the distribution fits the data well on the
tail. We will use maximum likelihood for this purpose. We have already
seen that the loss function is different from the Kullback-Leibler
divergence that the MLE estimate attempts to optimise, so the MLE is
not optimal in terms of minimising our expected loss function, and
further work could go into devising better estimation methods for the
misspecified prior case. For the MLE estimation, the details in each
case are presented here:

\subsubsection{Normal Distribution}

We have $n$ samples which we model as having mean $\theta_i$ following
a normal distribution with mean 0 and variance $\tau^2$, and each
observation $x_i$ following a normal distribution with mean $\theta_i$
and variance $\sigma_i{}^2$. We want to maximise the log-likelihood of
all the data points with $\theta_i>a$ for some cuttoff $a$. To
simplify this procedure, we will maximise the log-likelihood of all
data points for which $x_i>a$. The log-likelihood is then written
$$\sum_{x_i>a}
\left(-\frac{x_i{}^2}{2(\tau^2+\sigma_i{}^2)}-\frac{\log(\tau^2+\sigma_i{}^2)}{2}-\log\left(1-\Phi\left(\frac{a}{\sqrt{\tau^2+\sigma_i{}^2}}\right)\right)\right)$$
(The last term is because we must take the conditional log-likelihood
conditional on $x_i>a$.) Setting the derivative with respect to $\tau$ to zero, we get
$$\sum_{x_i>a}\left(
\frac{\tau
  x_i{}^2}{(\tau^2+\sigma_i{}^2)^2}-\frac{\tau}{(\tau^2+\sigma_i{}^2)}-\frac{\tau
  ae^{-\frac{a^2}{2(\tau^2+\sigma_i{}^2)}}}{\sqrt{2\pi}(\tau^2+\sigma_i{}^2)^{\frac{3}{2}}\left(1-\Phi\left(\frac{a}{\sqrt{\tau^2+\sigma_i{}^2}}\right)\right)}\right)=0$$

We can solve this numerically using Newton's method. We can use the
following method to obtain a good starting value. 
Since $a$ is reasonably large compared to $\tau$, we can approximate 
$$\frac{
  e^{-\frac{a^2}{2(\tau^2+\sigma_i{}^2)}}}{\sqrt{2\pi(\tau^2+\sigma_i{}^2)}\left(1-\Phi\left(\frac{a}{\sqrt{\tau^2+\sigma_i{}^2}}\right)\right)}\approx
\frac{a}{\tau^2+\sigma_i{}^2}$$ [NOTE: this is a poor
  approximation. Using it gives fairly bad approximations for
  $\hat{\tau}$. The approximations for other estimating priors later
  are better.]  so that the final term in the derivative of the
log-likelihood is approximately
$$\frac{\tau a^2}{(\tau^2+\sigma_i{}^2)^{2}}$$
We have assumed that
$\sigma_i$ is small compared to $\tau$, so we can set 

\begin{align*}
\sum_{x_i>a}\left(
\frac{\tau
  (x_i{}^2-a^2)}{(\tau^2+\sigma_i{}^2)^2}-\frac{\tau}{(\tau^2+\sigma_i{}^2)}\right)
&=\sum_{x_i>a}\left(\tau^{-3}
  (x_i{}^2-a^2)\left(1+\frac{\sigma_i{}^2}{\tau^2}\right)^{-2}
-\tau^{-1}\left(1+\frac{\sigma_i{}^2}{\tau^2}\right)^{-1}\right)\\
&\approx \sum_{x_i>a}\left(\tau^{-3}
  (x_i{}^2-a^2)\left(1-2\frac{\sigma_i{}^2}{\tau^2}\right)
-\tau^{-1}\left(1-\frac{\sigma_i{}^2}{\tau^2}\right)\right)\\
&= \tau^{-5}\sum_{x_i>a}\left(-\tau^{4}+\tau^2\left((x_i{}^2-a^2)+\sigma_i{}^2\right)-
 2\sigma_i{}^2(x_i{}^2-a^2)\right)\\
&= \tau^{-5}\left(-n_a\tau^{4}+\tau^2\sum_{x_i>a}\left((x_i{}^2-a^2)+\sigma_i{}^2\right)-
 2\sum_{x_i>a}\sigma_i{}^2(x_i{}^2-a^2)\right)
\end{align*}
where $n_a$ is the number of points with $x_i>a$.

 We solve for when this is equal to zero using the quadratic formula to
get:
$$\tau^2\approx\frac{\sum_{x_i>a}\left((x_i{}^2-a^2)+\sigma_i{}^2\right)+\sqrt{\left(\sum_{x_i>a}\left((x_i{}^2-a^2)+\sigma_i{}^2\right)\right)^2-8n_a\sum_{x_i>a}(x_i{}^2-a^2)\sigma_i^2}}{2n_a}$$
which should give an approximation to the true value of $\tau$. If
we further make the approximation that
$\frac{8n_a\sum_{x_i>a}(x_i{}^2-a^2)\sigma_i^2}{\left(\sum_{x_i>a}\left((x_i{}^2-a^2)+\sigma_i{}^2\right)\right)^2}$
is small, then we have 
$$\sqrt{\left(\sum_{x_i>a}\left((x_i{}^2-a^2)+\sigma_i{}^2\right)\right)^2-8n_a\sum_{x_i>a}(x_i{}^2-a^2)\sigma_i^2}\approx \sum_{x_i>a}\left((x_i{}^2-a^2)+\sigma_i{}^2\right)-\frac{8n_a\sum_{x_i>a}(x_i{}^2-a^2)\sigma_i^2}{2\sum_{x_i>a}\left((x_i{}^2-a^2)+\sigma_i{}^2\right)}$$
which gives us
$$\tau^2\approx\frac{\sum_{x_i>a}\left((x_i{}^2-a^2)+\sigma_i{}^2\right)}{n_a}-\frac{2\sum_{x_i>a}(x_i{}^2-a^2)\sigma_i^2}{\sum_{x_i>a}\left((x_i{}^2-a^2)+\sigma_i{}^2\right)}$$

We can compare this approximate MLE estimate of $\tau^2$ to the
theoretically best estimate for the exponential and Pareto cases. If
we assume that $\sigma_i$ are all small, then the term
$\frac{\sum_{x_i>a}(x_i{}^2-a^2)\sigma_i^2}{\sum_{x_i>a}\left((x_i{}^2-a^2)+\sigma_i{}^2\right)}$
is approximately $\sum_{x_i>a}\frac{(x_i{}^2-a^2)}{\sum_{x_i>a}
  (x_i{}^2-a^2)}\sigma_i^2$, which is a weighted mean of the
$\sigma_i{}^2$. Therefore the expected value is the expected value of
$\sigma_i{}^2$, so we have 
$${\mathbb E}(\hat{\tau}^2)\approx{\mathbb E}((x_i{}^2-a^2))-{\mathbb E}(\sigma_i{}^2)$$
Since $x_i$ is normally distributed with mean $\theta_i$ and variance
$\sigma_i{}^2$, we have that $${\mathbb
  E}(x_i{}^2|\theta_i)=\left({\mathbb
  E}(x_i|\theta_i)\right)^2+\sigma_i{}^2=\theta_i{}^2+\sigma_i{}^2$$
Therefore we have 
$${\mathbb E}(\hat{\tau}^2)\approx{\mathbb E_{x_i>a}}(\theta_i{}^2-a^2)\approx{\mathbb E_{\theta_i>a}}(\theta_i{}^2-a^2)$$
For the exponential true prior, we have that conditional on
$\theta_i>a$, we have $T=\theta_i-a$ follows an exponential
distribution with $\lambda=1$ and $a=\log(10)$. Therefore
\begin{align*}
{\mathbb E_{\theta_i>a}}(\theta_i{}^2)&={\mathbb
  E}((T+a)^2)=a^2+2a{\mathbb E}(T)+{\mathbb E}(T^2)
=a^2+2\frac{a}{\lambda}+\frac{2}{\lambda^2}\\
{\mathbb E_{\theta_i>a}}(\theta_i{}^2)-a^2&=2\frac{a}{\lambda}+\frac{2}{\lambda^2}=2\log(10)+2= 6.60517\\
\end{align*}
Therefore, for a large sample 
$$\hat{\tau}\approx\sqrt{6.60517}=2.570053$$
This is quite far from the optimal estimate of $1.700526$.

For the Pareto true prior, the variance is infinite (since
$\alpha\leqslant 2$), so the distribution of the MLE $\hat{\tau}^2$
has infinite mean. This means we cannot apply the law of large numbers
to assert that for large sample size $\hat{\tau}^2$ will converge in
distribution to a constant. More specifically, $\theta_i{}^2$ follows
a Pareto distribution with $\alpha=1$ and $\eta=\frac{1}{4}$.
The sum of Pareto distributions with small $\alpha$ is approximately
equal to the maximum value, which has distribution function 
$$F_{\sum\theta_i{}^2}(x)=\left(1-\frac{\eta}{x}\right)^{n}$$
We also have 
$$F_{\hat{\tau}^2}(x)=F_{n_a\hat{\tau}^2}(n_ax)=F_{\sum\theta_i{}^2}(n_a(x+a^2))=\left(1-\frac{\eta}{n_a(x+a^2)}\right)^{n}\approx
e^{-\frac{\eta n}{\left(x+a^2\right)n_a}}$$
We are interested in $\frac{1}{\hat{\tau}^2}$, because this is the value
that is important for our posterior mean estimate. The survival
function of $\frac{1}{\hat{\tau}^2}$ is 
$$S_{\frac{1}{\hat{\tau}^2}}(x)=F_{\hat{\tau}^2}\left(\frac{1}{x}\right)\approx
e^{-\frac{\eta nx}{n_a\left(1+a^2x\right)}}$$ That is, $\frac{1}{\hat{\tau}^2}$ approximately follows
an exponential distribution with parameter
$\frac{1}{4P(\theta_i>a)}=2.5$. This can be quite different from the optimal $0.72$. Indeed we get 
\begin{align*}
{\mathbb
  E}\left(\left(\frac{1}{\hat{\tau}^2}-0.72\right)^2\right)&={\mathbb
  E}\left(\left(\frac{1}{\hat{\tau}^2}-0.4\right)^2\right)+(0.4-0.72)^2=\Var\left(\frac{1}{\hat{\tau}^2}\right)+0.32^2=0.16+0.1024=0.2624\\
\end{align*}

Meaning that the MLE estimate for $\tau^2$ does not give a good estimate.


\subsubsection{Exponential Distribution}

The likelihood of a point $(x_i,\sigma_i)$ is 
\begin{align*}
\int_0^\infty \lambda e^{-\lambda
  \theta}\frac{e^{-\frac{(x_i-\theta)^2}{2\sigma_i{}^2}}}{\sqrt{2\pi}\sigma_i}\,d\theta&=
\frac{\lambda e^{\frac{\lambda^2\sigma_i{}^2}{2}-\lambda
    x_i}}{\sqrt{2\pi}\sigma_i}\int_0^\infty
e^{-\frac{(\theta+\lambda\sigma_i{}^2-x_i)^2}{2\sigma_i{}^2}}\,d\theta\\
&=
\lambda e^{\frac{\lambda^2\sigma_i{}^2}{2}-\lambda x_i}\Phi\left(\frac{x_i}{\sigma_i}-\lambda\sigma_i\right)\\
\end{align*}

Since $x_i>a$ and $\sigma_i$ is small, we can approximate 
$$\Phi\left(\frac{x_i}{\sigma_i}-\lambda\sigma_i\right)\approx 1$$
so the log-likelihood is approximately
$$\sum_{x_i>a}\left(\log(\lambda)+\frac{\lambda^2\sigma_i{}^2}{2}-\lambda
x_i\right)=\lambda^2 \sum_{x_i>a}\frac{\sigma_i{}^2}{2}-\lambda
\sum_{x_i>a} x_i+n_a\log(\lambda)$$
However, we want the conditional log-likelihood given $x_i>a$. Since
$\sigma_i$ is small, we will set this approximately equal to the
likelihood conditional on $\theta_i>a$, which is 
$$\sum_{x_i>a}\left(\log(\lambda)+\frac{\lambda^2\sigma_i{}^2}{2}-\lambda
x_i\right)=\lambda^2 \sum_{x_i>a}\frac{\sigma_i{}^2}{2}-\lambda
\sum_{x_i>a} (x_i-a)+n_a\log(\lambda)$$
Setting the derivative with respect to $\lambda$ to zero gives 
\begin{align*}
\lambda \sum_{x_i>a}\sigma_i{}^2-\sum_{x_i>a}
(x_i-a)+\frac{n_a}{\lambda}&=0\\
\lambda^2 \sum_{x_i>a}\sigma_i{}^2-\lambda\sum_{x_i>a}
(x_i-a)+n_a&=0\\
\lambda&=\frac{\sum_{x_i>a}(x_i-a)\pm\sqrt{\left(\sum_{x_i>a}(x_i-a)\right)^2-4n_a\sum_{x_i>a}\sigma_i{}^2}}{2\sum_{x_i>a}\sigma_i{}^2}
\end{align*}
so the log-likelihood is maximised by 
$$\lambda=\frac{\sum_{x_i>a}(x_i-a)-\sqrt{\left(\sum_{x_i>a}(x_i-a)\right)^2-4n_a\sum_{x_i>a}\sigma_i{}^2}}{2\sum_{x_i>a}\sigma_i{}^2}$$
(the other zero is because the
approximation $$\Phi\left(\frac{x_i}{\sigma_i}-\lambda\sigma_i\right)\approx
1$$ does not hold for $\lambda\approx \frac{\sum_{x_i>a}x_i}{\sum_{x_i>a}\sigma_i{}^2}$)
Since $\sigma_i$ is small, we can approximate 
$$\sqrt{\left(\sum_{x_i>a}(x_i-a)\right)^2-4n_a\sum_{x_i>a}\sigma_i{}^2}\approx \sum_{x_i>a}(x_i-a)-\frac{4n_a\sum_{x_i>a}\sigma_i{}^2}{2\sum_{x_i>a}(x_i-a)}$$
Which gives 
$$\hat{\lambda}\approx\frac{4n_a\sum_{x_i>a}\sigma_i{}^2}{4\left(\sum_{x_i>a}(x_i-a)\right)\left(\sum_{x_i>a}\sigma_i{}^2\right)}=\frac{n_a}{\sum_{x_i>a}(x_i-a)}$$

When the true prior is normal, we see that
${\mathbb E}(x_i-a|x_i>a)$ is the mean of a truncated normal
distribution, and is given by $$\tau\frac{e^{-\frac{a^2}{2\tau^2}}}{\sqrt{2\pi}\left(1-\Phi\left(\frac{a}{\tau}\right)\right)}-a$$
Substituting $\tau=1$ and $\Phi(a)=0.9$, we get that 
${\mathbb E}(x_i-a|x_i>a)=\frac{e^{-\frac{1.281552^2}{2}}}{0.1\sqrt{2\pi}}-a=1.754982-1.281552=0.4734308$
Therefore, for a large sample, our estimate $\hat{\lambda}$ will
converge to $\frac{1}{0.4734308}=2.112241$.

For the Pareto true prior, we have ${\mathbb E}(x_i|x_i>a)=2a$.
Despite the fact that the variance is infinite, the law of large
numbers still ensures that the sample mean of the $x_i$ does converge to
$2a$ as sample size tends to infinity. We can therefore substitute
$2a$ for this sum in the expression to get
$$\hat{\lambda}\approx\frac{n_a}{\sum_{x_i>a}(x_i-a)}=\frac{1}{a}=\frac{2}{\sqrt{10}}=0.6324555$$

\subsubsection{Pareto Distribution}

For the Pareto estimating prior, the likelihood of $\theta_i$
is $$\alpha\frac{\eta^\alpha}{\theta_i{}^{\alpha+1}}$$ and the probability
of a value exceeding $a$ is $\frac{\eta^\alpha}{a^\alpha}$. The
likelihood of $(x_i,\sigma_i)$ is therefore 
$$\int_{\eta}^\infty \alpha\frac{\eta^\alpha}{\theta^{\alpha+1}\sqrt{2\pi}\sigma_i}e^{-\frac{(x_i-\theta)^2}{2\sigma_i{}^2}}\,d\theta$$
Letting $\frac{\theta}{x_i}=1+\xi$, this integral becomes 
\begin{align*}
\alpha\frac{\eta^\alpha}{x_i^{\alpha+1}\sqrt{2\pi}\sigma_i}\int_{\eta}^\infty
(1+\xi)^{-\alpha-1}e^{-\frac{x_i{}^2\xi^2}{2\sigma_i{}^2}}\,d\theta&=\alpha\frac{\eta^\alpha}{x_i^{\alpha}\sqrt{2\pi}\sigma_i}\int_{\frac{\eta}{x_i}-1}^\infty
(1+\xi)^{-\alpha-1}e^{-\frac{x_i{}^2\xi^2}{2\sigma_i{}^2}}\,d\xi\\
&=\alpha\frac{\eta^\alpha}{x_i^{\alpha+1}}\int_{\frac{\eta}{x_i}-1}^\infty
\left(1-(\alpha+1)\xi+\frac{(\alpha+1)(\alpha+2)}{2}\xi^2-\cdots\right)\frac{x_ie^{-\frac{x_i{}^2\xi^2}{2\sigma_i{}^2}}}{\sqrt{2\pi}\sigma_i}\,d\xi\\
&\approx\alpha\frac{\eta^\alpha}{x_i^{\alpha+1}}{\mathbb E}_{\xi\sim N\left(0,\frac{\sigma_i{}^2}{x_i{}^2}\right)}
\left(1-(\alpha+1)\xi+\frac{(\alpha+1)(\alpha+2)}{2}\xi^2-\cdots\right)\\
&\approx\alpha\frac{\eta^\alpha}{x_i^{\alpha+1}}
\left(1+\frac{(\alpha+1)(\alpha+2)\sigma_i{}^2}{2x_i{}^2}+\cdots\right)\\
\end{align*}

so the
conditional likelihood of $x_i$ given that $\theta_i>a$ is
approximately
$$\alpha\frac{a^\alpha}{x_i^{\alpha+1}}
\left(1+\frac{(\alpha+1)(\alpha+2)\sigma_i{}^2}{2x_i{}^2}\right)$$
The conditional log-likelihood is therefore
$$\alpha\log(a)-(\alpha+1)\log(x_i)+\log(\alpha)+\log\left(1+\frac{(\alpha+1)(\alpha+2)\sigma_i{}^2}{2x_i{}^2}\right)$$
Setting the derivative with respect to $\alpha$ to zero gives
\begin{align*}
\sum_{x_i>a}\left(\log(a)-\log(x_i)+\frac{1}{\alpha}+\frac{(2\alpha+3)\sigma_i{}^2}{2x_i{}^2\left(1+\frac{(\alpha+1)(\alpha+2)\sigma_i{}^2}{2x_i{}^2}\right)}\right)&=0\\
\sum_{x_i>a}\left(\log(a)-\log(x_i)+\frac{1}{\alpha}+\frac{(2\alpha+3)\sigma_i{}^2}{2x_i{}^2+(\alpha+1)(\alpha+2)\sigma_i{}^2}\right)&=0\\
\frac{n_a}{\alpha}+\sum_{x_i>a}\left(\log(a)-\log(x_i)\right)+(2\alpha+3)\sum_{x_i>a}\frac{\sigma_i{}^2}{2x_i{}^2}-(\alpha+1)(\alpha+2)(2\alpha+3)\sum_{x_i>a}\frac{\sigma_i{}^4}{4x_i{}^4}&=0\\
n_a+\alpha\sum_{x_i>a}\left(\log(a)-\log(x_i)+3\frac{\sigma_i{}^2}{2x_i{}^2}-6\frac{\sigma_i{}^4}{4x_i{}^4}\right)+\alpha^2\sum_{x_i>a}\left(\frac{\sigma_i{}^2}{x_i{}^2}-13\frac{\sigma_i{}^4}{4x_i{}^4}\right)-9\alpha^3\sum_{x_i>a}\frac{\sigma_i{}^4}{4x_i{}^4}-\alpha^4\sum_{x_i>a}\frac{\sigma_i{}^4}{2x_i{}^4}&=0\\
\end{align*}

$$\approx \alpha\log(a)-(\alpha+1)\log(x_i)+\log(\alpha)+\frac{(\alpha+1)(\alpha+2)\sigma_i{}^2}{2x_i{}^2}$$
Setting the derivative with respect to $\alpha$ to zero gives
\begin{align*}
\sum_{x_i>a}\left(\log(a)-\log(x_i)+\frac{1}{\alpha}+\frac{(2\alpha+3)\sigma_i{}^2}{2x_i{}^2}\right)&=0\\
\frac{n_a}{\alpha}+\sum_{x_i>a}\left(\log(a)-\log(x_i)\right)+(2\alpha+3)\sum_{x_i>a}\frac{\sigma_i{}^2}{2x_i{}^2}&=0\\
n_a+\alpha\sum_{x_i>a}\left(\log(a)-\log(x_i)+3\frac{\sigma_i{}^2}{2x_i{}^2}\right)+\alpha^2\sum_{x_i>a}\frac{\sigma_i{}^2}{x_i{}^2}&=0\\
\end{align*}
Which has solution
\begin{align*}
\alpha&=\frac{\sum_{x_i>a}\left(\log(x_i)-\log(a)-3\frac{\sigma_i{}^2}{2x_i{}^2}\right)\pm\sqrt{\left(\sum_{x_i>a}\left(\log(x_i)-\log(a)-3\frac{\sigma_i{}^2}{2x_i{}^2}\right)\right)^2-4n_a\sum_{x_i>a}\frac{\sigma_i{}^2}{x_i{}^2}}}{2\sum_{x_i>a}\frac{\sigma_i{}^2}{x_i{}^2}}\\
\end{align*}
Assuming that $\frac{\sigma_i{}^2}{x_i{}^2}$ is small. we have the
approximation
\begin{align*}
&\sqrt{\left(\sum_{x_i>a}\left(\log(x_i)-\log(a)-3\frac{\sigma_i{}^2}{2x_i{}^2}\right)\right)^2-4n_a\sum_{x_i>a}\frac{\sigma_i{}^2}{x_i{}^2}}\\
\approx&\sum_{x_i>a}\left(\log(x_i)-\log(a)-3\frac{\sigma_i{}^2}{2x_i{}^2}\right)-\frac{2n_a\sum_{x_i>a}\frac{\sigma_i{}^2}{x_i{}^2}}{\sum_{x_i>a}\left(\log(x_i)-\log(a)-3\frac{\sigma_i{}^2}{2x_i{}^2}\right)}-\frac{2n_a{}^2\left(\sum_{x_i>a}\frac{\sigma_i{}^2}{x_i{}^2}\right)^2}{\left(\sum_{x_i>a}\left(\log(x_i)-\log(a)-3\frac{\sigma_i{}^2}{2x_i{}^2}\right)\right)^3}\\
\end{align*}
which gives the MLE 
\begin{align*}
\alpha&=\frac{n_a\sum_{x_i>a}\frac{\sigma_i{}^2}{x_i{}^2}}{\left(\sum_{x_i>a}\frac{\sigma_i{}^2}{x_i{}^2}\right)\left(\sum_{x_i>a}\left(\log(x_i)-\log(a)-3\frac{\sigma_i{}^2}{2x_i{}^2}\right)\right)}+\frac{n_a{}^2\left(\sum_{x_i>a}\frac{\sigma_i{}^2}{x_i{}^2}\right)^2}{2\left(\sum_{x_i>a}\frac{\sigma_i{}^2}{x_i{}^2}\right)\left(\sum_{x_i>a}\left(\log(x_i)-\log(a)-3\frac{\sigma_i{}^2}{2x_i{}^2}\right)\right)^3}\\
&=\frac{n_a}{\sum_{x_i>a}\left(\log(x_i)-\log(a)-3\frac{\sigma_i{}^2}{2x_i{}^2}\right)}+\frac{n_a{}^2\left(\sum_{x_i>a}\frac{\sigma_i{}^2}{x_i{}^2}\right)}{\left(\sum_{x_i>a}\left(\log(x_i)-\log(a)-3\frac{\sigma_i{}^2}{2x_i{}^2}\right)\right)^3}\\
&\approx\frac{n_a}{\sum_{x_i>a}\left(\log(x_i)-\log(a)\right)}+\left(\sum_{x_i>a}\frac{\sigma_i{}^2}{2x_i{}^2}\right)\left(\frac{3n_a}{\left(\sum_{x_i>a}\left(\log(x_i)-\log(a)\right)\right)^2}+\frac{n_a{}^2}{\left(\sum_{x_i>a}\left(\log(x_i)-\log(a)\right)\right)^3}\right)\\
\end{align*}

For our specific case, the normal true prior has $\tau=1$ and
$a=\Phi^{-1}(0.9)=1.281552$. Empirically, for these parameters,
${\mathbb E}(\log(X))=0.538$, so ${\mathbb
  E}(\log(X))-\log(a)=0.29$, and ${\mathbb
  E}\left(\frac{1}{X^2}\right)=0.37$. Therefore, the expected value of
$\hat{\alpha}$ is  
$${\mathbb E}(\hat{\alpha})=\frac{1}{0.29}+0.3694015{\mathbb E}{\sigma_i{}^2}\left(\frac{3}{0.29^2}+\frac{1}{0.39^3}\right)=3.45+28.32{\mathbb E}{\sigma_i{}^2}$$ 
Since $\sigma_i$ follows an exponential distribution with
$\lambda=50$, so ${\mathbb E}(\sigma_i{}^2)=\frac{2}{50^2}=0.0008$,
which means that 
$${\mathbb E}(\hat{\alpha})=3.45+0.02=3.47$$

For the exponential true prior, we have
\begin{proposition}
If $X$ follows an exponential distribution with rate $\lambda$, then
the function
$f(\lambda)={\mathbb E}(\log(1+X))$ satisfies the differential
equation
$$f'(\lambda)=f(\lambda)-\frac{1}{\lambda}$$
\end{proposition}

\begin{proof}
We have $f(\lambda)=\int_0^\infty \lambda e^{-\lambda
  x}\log(1+x)\,dx$. This gives 
\begin{align*}
f'(\lambda)&=\int_0^\infty  e^{-\lambda
  x}\log(1+x)\,dx-\int_0^\infty  \lambda xe^{-\lambda
  x}\log(1+x)\,dx\\
&=\int_0^\infty  e^{-\lambda
  x}\log(1+x)\,dx-\left[-e^{-\lambda x}x\log(1+x)\right]_0^\infty-\int_0^\infty  \left(\log(1+x)+\frac{x}{1+x}\right)e^{-\lambda
  x}\,dx\\
&=-\int_0^\infty  \frac{x}{1+x}e^{-\lambda
  x}\,dx\\
&=-\int_0^\infty  \left(1-\frac{1}{1+x}\right)e^{-\lambda
  x}\,dx\\
&=\int_0^\infty  \frac{e^{-\lambda  x}}{1+x}\,dx-\frac{1}{\lambda}\\
\end{align*}
On the other hand, integration by parts gives 
\begin{align*}
f(\lambda)&=\int_0^\infty \lambda e^{-\lambda
  x}\log(1+x)\,dx\\
&=\left[-e^{-\lambda
  x}\log(1+x)\right]_0^\infty+\int_0^\infty  \frac{1}{1+x}e^{-\lambda
  x}\,dx\\
&=\int_0^\infty  \frac{1}{1+x}e^{-\lambda
  x}\,dx\\
\end{align*}
Substituting this into the previous equation gives 
$$f'(\lambda)=f(\lambda)-\frac{1}{\lambda}$$
\end{proof}

\begin{proposition}
If $X$ follows an exponential distribution with rate $\lambda$, then
the function $g(\lambda)={\mathbb E}\left(\frac{1}{(1+X)^2}\right)$
satisfies the differential equation
$$g'(\lambda)=\left(1+\frac{2}{\lambda}\right)g(\lambda)-1$$
\end{proposition}

\begin{proof}
We have $g(\lambda)=\int_0^\infty \frac{\lambda e^{-\lambda
  x}}{(1+x)^2}\,dx$. This gives 
\begin{align*}
g'(\lambda)&=\int_0^\infty  e^{-\lambda
  x}\frac{1-\lambda x}{(1+x)^2}\,dx\\
&=(1+\lambda)\int_0^\infty  \frac{e^{-\lambda
  x}}{(1+x)^2}\,dx-\lambda \int_0^\infty  \frac{e^{-\lambda
  x}}{(1+x)}\,dx \\
\end{align*}
On the other hand, integration by parts gives 
\begin{align*}
\int_0^\infty \frac{\lambda e^{-\lambda
  x}}{1+x}\,dx&=\left[-\frac{e^{-\lambda
  x}}{1+x}\right]_0^\infty-\int_0^\infty  \frac{e^{-\lambda
  x}}{(1+x)^2}\,dx\\
&=1-\frac{g(\lambda)}{\lambda}\\
\end{align*}
This gives us 
\begin{align*}
g'(\lambda)&=\frac{(1+\lambda)}{\lambda}g(\lambda)-\left(1-\frac{g(\lambda)}{\lambda}\right)\\
&=\left(1+\frac{2}{\lambda}\right)g(\lambda)-1\\
\end{align*}
\end{proof}

This means that for an exponential with parameter $\lambda=1$ and
cut-off $a=\log(10)$, $Z=\frac{X}{a}-1$ follows an exponential
distribution with rate $a$, so $\log(X)-\log(a)=\log(1+Z)$,
so its expected value is $f\left(a\right)$, where $f$ is the
solution to $$f'(\lambda)=f(\lambda)-\frac{1}{\lambda}$$
Similarly, $X_i{}^{-2}=(a(1+Z))^{-2}$, so ${\mathbb E}(X_i{}^{-2})=a^{-2}g\left(a\right)$.
The expected value of $\hat{\alpha}$ is then
$$\frac{1}{f\left(\log(10)\right)}+\frac{{\mathbb E}(\sigma_i{}^2)g\left(\log(10)\right)}{2\log(10)^2}\left(\frac{3}{f\left(\log(10)\right)^2}+\frac{1}{f\left(\log(10)\right)^3}\right)$$
Numerically, we find $f\left(\log(10)\right)=0.3239$ and
$g\left(\log(10)\right)=0.5853$. Substituting these
values into the equation gives 
$$\hat{\alpha}=\frac{1}{0.3239}+{\mathbb
  E}(\sigma_i{}^2)\frac{0.5853}{2\log(10)^2}\left(\frac{3}{0.3239^2}+\frac{1}{0.3239^3}\right)= 3.087+3.203{\mathbb
  E}(\sigma_i{}^2)\approx 3.151$$

\end{document}